\documentclass{vldb}
\pdfoutput=1
\usepackage{graphicx}
\usepackage{graphics}
\usepackage{algorithm}
\usepackage{algorithmic}
\usepackage{amsmath,wrapfig, multirow,float,caption}
\usepackage{subcaption}
\newcommand{\eat}[1]{}
\usepackage{colortbl}
\definecolor{Brown}{cmyk}{0, 0.8, 1, 0.6}

\DeclareCaptionType{copyrightbox}
\newdef{theorem}{Theorem}
\newdef{problem}{Problem definition}
\newtheorem{lemma}{Lemma}
\newtheorem{definition}{Definition}
\newtheorem{example}{Running Example}
\newcommand{\new}[1]{{\color{blue}#1}}

\begin{document}
\title{Optimization in Knowledge-Intensive Crowdsourcing}

\newcommand{\sys}{{\sc{SmartCrowd}}}
\newcommand{\userindex}{{\sc{VirtualWorker}}}
\newcommand{\Index}{{\sc{C-dex}}}
\newcommand{\vindex}{{\sc{C-dex}$^{+}$}}

\author{
\alignauthor Senjuti Basu Roy$^{\dag}$, Ioanna Lykourentzou$^{\dag\dag}$,
 Saravanan Thirumuruganathan$^{\ddag,\triangle}$ \\
 Sihem Amer-Yahia$^{\diamond}$, Gautam Das$^{\ddag,\triangle}$.  \\
\affaddr{
$^{\dag}$UW Tacoma,
$^{\dag\dag}$CRP Henri Tudor/INRIA Nancy Grand-Est,
$^{\ddag}$UT Arlington,
$^{\triangle}$QCRI,
$^{\diamond}$ CNRS, LIG
}
{\email{
	senjutib@uw.edu, 
	ioanna.lykourentzou@\{tudor.lu,inria.fr\},
	saravanan.thirumuruganathan@mavs.uta.edu, 
    sihem.amer-yahia@imag.fr,	
	gdas@uta.edu}
}
}

\maketitle

\pagestyle{plain}
\pagenumbering{arabic}
\maketitle
\begin{abstract}
We present \sys, a framework for optimizing collaborative
knowledge-intensive crowdsourcing. \sys\ distinguishes itself by
accounting for human factors in the process of assigning tasks to
workers. Human factors designate workers' expertise in different
skills, their expected minimum wage, and their availability. In \sys,
we formulate task assignment as an optimization problem, and rely on
pre-indexing workers and maintaining the indexes adaptively, in such a
way that the task assignment process gets optimized both
qualitatively, and computation time-wise. We present rigorous theoretical analyses of the  optimization problem and propose optimal and approximation algorithms. We finally perform extensive performance and quality experiments using real and synthetic data to 
demonstrate that adaptive indexing in \sys\ is necessary to achieve
efficient high quality task assignment.
\end{abstract}


\vspace{-0.15in}
\section{Introduction}
Knowledge-intensive crowdsourcing (KI-C) is acknowledged as one of the most promising  areas of next-generation crowdsourcing \cite{Kittur:2013:FCW:2441776.2441923}, mostly for the critical role it can play in todays knowledge-savvy economy. KI-C refers to the collaborative creation of knowledge content (for example Wikipedia articles, or news articles) through crowdsourcing. Crowd workers, each having a certain degree of expertise, collaborate and ``build'' on each other's contributions to gradually increase the quality of each knowledge piece (hereby referred to as ``task''). Despite its importance, no work or platform so far has tried to optimize KI-C, a fact which often results in poor task quality and undermines the reliability of crowds for knowledge intensive applications.

In this paper we propose \emph{\sys}, an optimization framework for knowledge-intensive collaborative crowdsourcing. \sys\ aims at improving KI-C by optimizing one of its fundamental processes, i.e., worker-to-task assignment, while taking into account the dynamic and uncertain nature of a real crowdsourcing environment \cite{DBLP:conf/dbcrowd/RoyLTAD13}. 

 
Consider the example of a KI-C application offering news articles on demand as a service to interested stakeholders, such as publication houses, blogs, individuals, etc. 
Several thousands of workers are potentially available to compose thousands of news articles collaboratively. It is easy to imagine that such an application needs to judiciously assign workers to tasks, so as to ensure high quality article delivery while being cost-effective. 
Two main challenges need to be investigated: 1) How to formalize the KI-C worker-to-task assignment problem?  2) How to solve the problem efficiently so as to warrant the desired quality/cost outcome of the KI-C platform, while taking into account the unpredictability of human behavior and the volatility of workers in a realistic crowdsourcing environment?

\sys\ has been envisioned as follows: 
 First, we formalize the {\bf KI-C worker-to-task assignment as an optimization problem} (Section \ref{dm}). In our formulation, the resources are the worker profiles (knowledge skill per domain, requested wage) and the tasks are the news articles (assumed to have a minimum quality, maximum cost and skills needed).\footnote{\small With the availability of historical information, worker profiles (knowledge skills and expected wage) can be learned by the platform. This complex profile learning problem is an independent research problem in its own merit, orthogonal to this work.} 
The objective function is formalized so as to guarantee that each task surpasses a certain quality threshold, stays below a cost limit, and that workers are not over or under utilized. Given the innate uncertainty induced by human involvement, we also use \emph{probabilistic modeling} to include one of the human factors (formalized as the workers' acceptance ratio\footnote{\small Acceptance ratio of a worker is the probability that she accepts a recommended task.}) in the problem formulation.

Then, we argue that it may be prohibitively expensive to assign workers to the tasks optimally
 in real time and reason about the necessity of \emph{pre-computation} for efficiency reasons. We propose {\bf index design (\Index) as a means to efficiently address the KI-C optimization problem} (Section \ref{smartcrowd}). One of the novel contributions of this work is in proposing how the \Index\ solution can be used to pre-compute {\em crowd indexes} for KI-C tasks, which can be used efficiently afterwards during the actual worker-to-task assignment process. 
We show how KI-C tasks 
could benefit from crowd-indexes to efficiently maximize the objective function.

Third, we examine {\bf the problem under dynamic conditions of the crowdsourcing environment}, where new workers may subscribe, existing ones may leave, worker profiles may change over time,
 and workers may accept or decline recommended tasks. To tackle such unforeseen scenarios, \sys\ proposes optimal
\emph{adaptive maintenance of the pre-computed indexes}, while enforcing the non-preemption of workers.\footnote{\small Non-preemption ensures that a worker cannot be interrupted after she is assigned to a task.}

Fourth, we prove {\bf several theoretical properties of the \Index\ design problem}, such as \emph{NP-Completeness} (using a reduction from the Multiple Knapsack Problem \cite{DBLP:books/fm/GareyJ79}), as well as \emph{sub-modularity} and \emph{monotonicity} under certain conditions. This in-depth theoretical analysis is critical to understand the problem complexity, as well as to design efficient principled solutions with theoretical guarantees. 

Finally, we propose {\bf novel optimal and approximate solutions for the index design and maintenance problem, depending on the exact problem conditions.} Our optimal solution uses an integer linear programming (ILP) approach (Section \ref{alg1}). For the case where optimal index building or maintenance is too expensive, we propose \emph{two efficient approximate strategies}: 1) a greedy computation and maintenance of \Index\ that needs polynomial computation time and admits a constant time approximation factor under certain conditions, and 2) \vindex, a strategy that is an optimized version of \Index, which leverages the clustering of similar workers (based on the notion of ``virtual worker'') to warrant further efficiency (Section \ref{approx}).

 

We design comprehensive experimental studies (Section~\ref{exp}) both with real-users and simulations to validate \sys, qualitatively and efficiency wise. With an appropriate and intelligent adaptation of Amazon Mechanical Turk (AMT), we conduct extensive quality experiment involving real workers to compose news articles. Such an adaptation is non-trivial and needs a careful design of the validation strategies, since AMT (or any other platforms) does not yet support KI-C tasks. 
Extensive simulation studies are used to further investigate our proposed framework, in terms of quality and efficiency. In these, we compare against several baseline algorithms, including one of the latest state-of-the-art techniques~\cite{chienJuHo} for online task assignment. 
The obtained results demonstrate that 
\emph{\Index\ and \vindex\ achieve 3x improvement}, both qualitatively and efficiency wise, corroborating the necessity of pre-computed indexes and their adaptive maintenance for the KI-C optimization problem. 
Our main contributions are summarized as follows:
\vspace{-0.05in}
\begin{enumerate}
\item We initiate the study of optimizing knowledge-intensive crowdsourcing (KI-C), formalize the problem, and propose rigorous theoretical analyses. 
\vspace{-0.05in}
\item We propose the necessity of index design and dynamic maintenance to address the KI-C optimization problem.  We propose novel optimal and approximate solutions (\Index, greedy \Index, and \vindex) for index creation as well as adaptive maintenance.
\vspace{-0.05in}
\item We conduct extensive experiments on real and simulated crowdsourcing settings to demonstrate the effectiveness of our proposed solution qualitatively and efficiency wise.
\end{enumerate}

\noindent Sections~\ref{dm} and \ref{smartcrowd} contain the settings, problem statements, and theoretical analyses. Sections~\ref{alg1} and \ref{approx} have the solutions. Sections~\ref{exp} and \ref{rel} contain the experiments and related work. We conclude in Section~\ref{conc}.
\section{KI-C Problem Settings} \label{dm}
\vspace{-0.1in}
\subsection{Data Model}
We are given a set of workers $\mathcal{U}=\{u_1,u_2,\ldots,u_n\}$, a
set of skills $\mathcal{S}=\{s_1,s_2, \ldots, s_m\}$ and a set of
tasks $T=\{t_1,t_2,\ldots,t_l\}$. In the context of
collaborative editing, skills represent topics such as Egyptian Politics, Play Station, or NSA document leakage. Tasks represent the documents that are being edited collaboratively. \\

\noindent{\bf Skills:} A skill is the knowledge on a particular topic
and is quantified in a continuous scale between $[0,1]$. It is
associated to workers and tasks. When associated to a worker, it
represents the worker's expertise of a topic. 
When associated to a task, a skill represents the minimum quality requirement for that task.
A value of $0$ for a skill reflects no expertise
of a worker for that skill. For a task, $0$ reflects no requirement
for that skill.\\


\noindent {\bf Workers:} Each worker $u \in \mathcal{U}$ has a {\em
  profile} that is  a vector, $\langle u_{s_1}, u_{s_2}, \ldots,
u_{s_m}, w_u, p_u \rangle$, of length $m+2$ describing her $m$ skills in
$\mathcal{S}$, her wage $w_u$, and her task acceptance ratio $p_u$. 
\vspace{-0.05in}
\begin{itemize}
  \item Skill $u_{s_i} \in [0,1]$ is the expertise level of worker $u$
    for skill $s_i$. Skill expertise reflects the quality that the
    worker's contribution will assign to a task accomplished by that
    worker. 
    \vspace{-0.05in}
  \item Wage $w_u \in [0,1]$ is the amount of money a worker $u$ is
    willing to accept to complete a task. The wage represents the
    minimum amount the worker expects to be paid for any task.
    \vspace{-0.05in}
  \item Acceptance ratio $p_u \in [0,1]$, the probability at which a
    worker $u$ accepts a task. It reflects the worker's willingness to
    complete tasks assigned to her. A value of $0$ is used to model
    workers who are not available (as workers who do not accept any
    task).
    \vspace{-0.05in}
\end{itemize}
\vspace{-0.05in}
We refer to a worker's skill, wage expectation and
acceptance ratio as {\em human factors} that may vary over the time. 
\eat{
\begin{table}
\begin{tabular} {|c|p{3cm}|}
\hline
  {\bf Notation} & {\bf Interpretation} \\
   \hline 
   $\mathcal{U}$ & the set of $n$ workers \\
   \hline 
   $\mathcal{S}$ & the set of $m$ skills \\
   \hline 
   $T$ & the set of $l$ tasks \\
   \hline
   $t = \langle Q_{t_1}, Q_{t_2}, \ldots, Q_{t_m}, W_t \rangle$ &  a task vector\\
   \hline
   $u = \langle u_{s_1}, u_{s_2}, \ldots, u_{s_m}, w_u, p_u \rangle$ & profile of worker $u$, $m$ skills, wage, and acceptance ratio \\
   \hline
   $i^{t} = \langle \mathcal{P}_i, \mathcal{L}_i \rangle$ & an \Index\ for task $t$ \\
   \hline
   $\mathcal{I}$ & a set of \Index \\
   \hline
   $\mathcal{I}_v$ & a set of \vindex \\
   \hline
   $v_j$ & value of task $j$ \\
   \hline
   $\mathcal{C}_u$ & the number of active tasks assigned to worker $u$ \\
   \hline
   $V$ & a Virtual Worker \\
   \hline
   $\mathcal{V}_u$ & the number of active tasks assigned to Virtual Worker $V$ \\
   \hline
  \end{tabular}
\caption{Notations used in the paper}\label{tab:notations}
\end{table}
}
\noindent {\bf Tasks:} A task $t \in T$ is a
vector, $\langle Q_{t_1}, Q_{t_2}, \ldots, Q_{t_m}, W_t\rangle$ of length $m$+1
reflecting its minimum skill requirement for each skill and its
maximum cost (or wage). A task $t$ that is being executed has a set of
contributors $\mathcal{U}_t \subseteq \mathcal{U}$ so far. For collaborative tasks, the quality of a task is the aggregate of the skill of the workers contributing to $t$, for a given skill. $t$ is hence characterized by:
\vspace{-0.05in}
\begin{itemize}
\item Current quality $q_{t_i}=\Sigma_{u \in \mathcal{U}_t} u_{s_i} \in
  [0,|\mathcal{U}_t|]$ for skill $s_i$, with $u_{s_i}$ being the expertise of
  worker $u$ on skill $s_i$.  $q_{t_i}$ aggregates the expertise of
  all workers who have contributed to $t$ so far.
  \vspace{-0.05in}
\item Current cost $w_t=\Sigma_{u \in \mathcal{U}_t} w_u \in [0,|\mathcal{U}_t|]$,
  with $w_u$ being the wage paid to worker $u$.  $w_t$ aggregates the
  wages of all workers who have contributed to $t$ so far.
  \vspace{-0.05in}
\end{itemize}
\vspace{-0.05in}
\eat{\new{We assume that each task $t$ requires one skill (for now) and is described by a pair $t = <Start, End, E_s, E_c>$, where $E_s$ is the expected skill required for $t$ and $E_c$ is its expected budget. {\em Start} and {\em End} correspond to the starting and ending time of the task. This allows us to model tasks that do not overlap and simultaneous tasks. 
We are also provided with past workload that provides some idea about task arrival properties. 
If no workload information is available, then we can assume that the skill levels required for a task are uniformly distributed.}}

\noindent {\bf Workload:} 
We assume a static workload $T$ that represents a set of active tasks over a time period.


\eat{
\begin{figure}[h]
\centering
\includegraphics[width=2.5in]{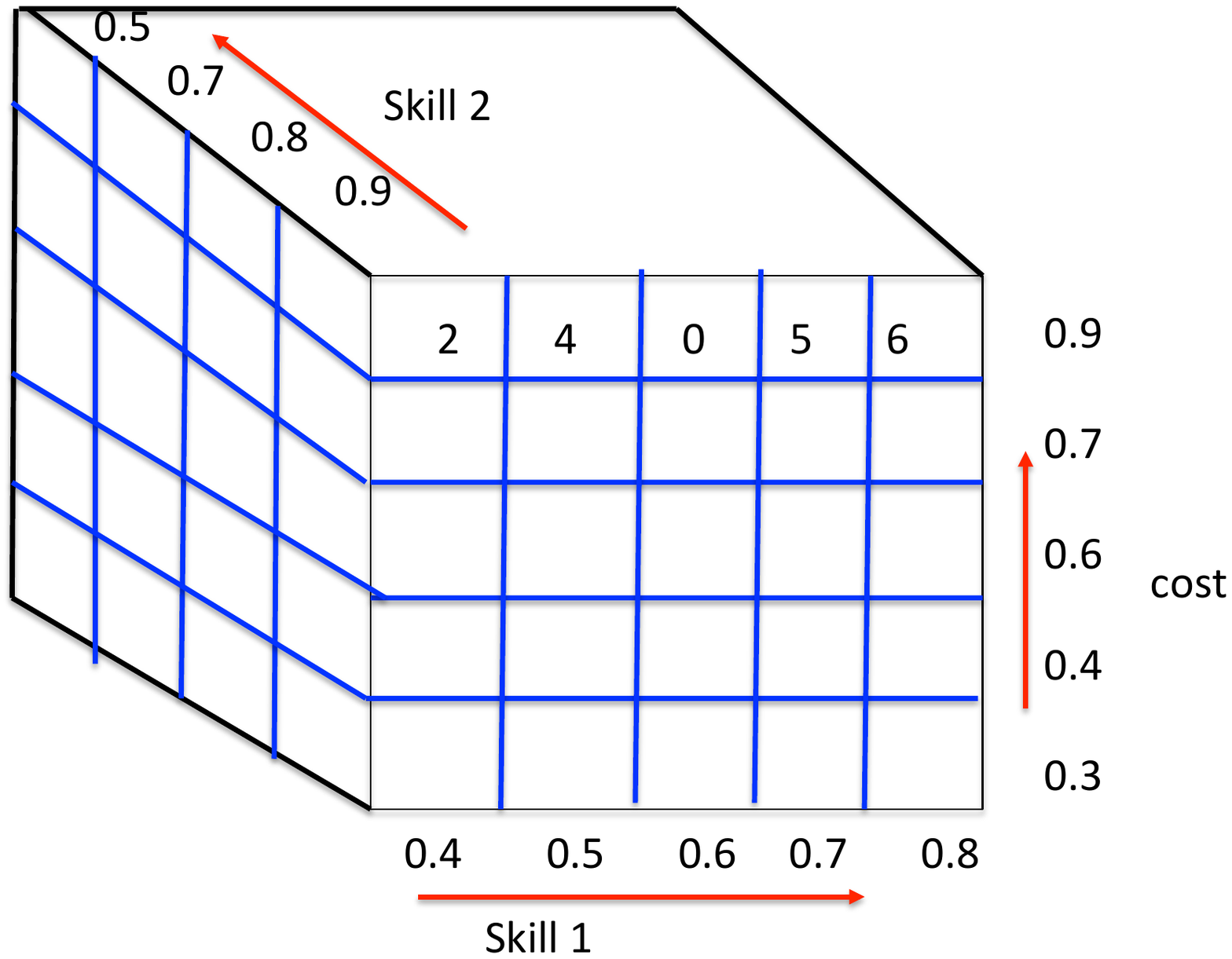}
\caption{\label{workload} When there are only two skills (skill-1 and
  skill-2) in the system, a workload $T$ is represented as a
  3-dimensional matrix. This partially complete workload describes
  that there are $2$ tasks with $\langle 0.4,0.9,0.9 \rangle$, $4$
  tasks with $\langle 0.5,0.9,0.9 \rangle$, $0$ tasks with $\langle
  0.6,0.9,0.9 \rangle$, $5$ tasks with $\langle 0.7,0.9,0.9 \rangle$,
  and $6$ tasks with $\langle 0.8,0.9,0.9 \rangle$.}
\end{figure}
}
\vspace{-0.1in}
\subsection{Constraints}
\label{constraints}
The following 
constraints are considered:
1. {\bf Quality constraint:} For each task $t \in T$, the worker-to-task assignment has to be such that the aggregated skill of assigned workers is at least as large as the minimum skill requirement of $t$ for each skill. \footnote{\small $Q_{t_j}$ is the threshold for skill $j$ and $q_{t_j} \geq Q_{t_j}$.}
2. {\bf Cost constraint:} For each task $t \in T$, the aggregated workers' wage ($w_t$) cannot exceed the maximum cost that $t$ can pay, i.e., $w_t \leq W_t$. 
3. {\bf Non-preemption constraint.} Once a worker has been assigned to a task, she cannot be pulled out of that task until finished. 
4. {\bf Tasks per worker constraint:} A worker must be assigned a minimum number of $X_l$ tasks and no more than $X_h$ tasks.

\subsection{Objective}
Given a set  $T$ of tasks and a set $\mathcal{U}$ of workers, the objective is to perform
worker-to-task assignment for all tasks in $T$, such that the overall task quality is maximized and the cost is minimized, while the constraints of skill, cost, and tasks-per-worker are satisfied. 

\vspace{-0.1in}
\begin{example}\label{runningex}
A running example is described consisting of a minuscule version of the news article composition task. Assume that the platform consists of 6 workers to compose 3 news articles (tasks) on ``Egyptian Politics''(t1), ``NSA leakage''(t2) and ``US Health Care Law'' (t3). For simplicity, we assume that all tasks belong to same topic (``Politics'') and therefore require only one skill (knowledge in politics). We also assume that $X_l=1$, $X_h=2$. Worker profiles (skill, wage,acceptance ratio) and task requirements (minimum quality, maximum cost) are depicted numerically in Tables \ref{tab:workers} and \ref{tab:tasks}. This example will be used throughout the paper to illustrate our solution.

\begin{table*}
\parbox{.35\linewidth}{
\begin{tabular}{| p{10em} | l | l | l | l | l | l |}
    \hline
    \textbf{Worker} & $\mathbf{u_1}$ & $\mathbf{u_2}$ & $\mathbf{u_3}$ & $\mathbf{u_4}$ & $\mathbf{u_5}$ & $\mathbf{u_6}$ \\ 
    \textbf{Skill} & 0.1 & 0.3  & 0.2 & 0.6 & 0.4 & 0.5\\
    \textbf{Wage} & 0.05 & 0.25  & 0.3 & 0.7 & 0.3 & 0.4\\ 
	\textbf{Acceptance ratio} & 0.8 & 0.7  & 0.8 & 0.5 & 0.6 & 0.9 \\   
    \hline
  \end{tabular}
\vspace{-0.1in} 
  \caption{Workers Profiles}
  \label{tab:workers}
}
\hfill
\parbox{.45\linewidth}{
\begin{tabular}{| p{10em} | p{2em} | p{2em} | p{2em} |}
    \hline
    \textbf{Task} & $\mathbf{t_1}$ & $\mathbf{t_2}$ & $\mathbf{t_3}$\\ 
    \textbf{Quality threshold} & 0.7 & 0.8  & 0.9 \\ 
    \textbf{Cost threshold} & 1.08 & 1.1 & 2.0 \\ 
    \hline
  \end{tabular}
\vspace{-0.1in}
\caption{Task Descriptions}\label{tab:tasks}
}
\end{table*}
\vspace{-0.1in}
\end{example}

\section{Smartcrowd}\label{smartcrowd}
The overall functionality of \sys\ is as follows: A set of indexes $\mathcal{I}$, referred to as \Index\ are pre-computed  based on a simple definition of past task {\em workload}. This step is referred to as the offline phase. Then, the idea is to make use of these indexes for efficient worker to task assignment once the actual tasks arrive. This latter step is construed as the online phase. However, indexing workers in KI-C  is more challenging than data indexing for query processing due to human factors in a dynamic environment, as workers may be unavailable/decline tasks, new workers may join, existing workers may have updated profile, etc. Therefore, \sys\ must design {\em adaptive maintenance strategies} of \Index\ to account for worker replacements, additions, deletions, or updates in their profiles. The maintenance strategies also need to be cognizant of workers non-preemption, since workers currently engaged in tasks can not be withdrawn until completion. 
\sys\ proposes further optimization opportunities for both offline and online phase, with a greedy \Index\ building and maintenance strategy and an alternative index, namely \vindex. The latter two strategies are crucial to ensure efficiency for applications that involve a large set of workers and tasks. Both of them are approximate, yet greedy \Index\ could entail provable approximation factor under certain conditions.
Interestingly, the treatment of index building or their adaptive maintenance is {\em uniform} inside \sys, with appropriate adaptation of similar solution strategy. Of course, if the actual tasks are substantially different from the workload, \sys\ has to halt and re-design the indexes from scratch. The latter scenario is orthogonal to us.

In the following we formalize the index design problem  (\Index) through which we aim at optimizing the KI-C optimization problem and propose an in-depth theoretical analysis. Then we investigate further optimization opportunities in index design by formalizing \vindex. Finally, we propose adaptive maintenance of the indexes, which is related to the online index maintenance phase.


\subsection{C-DEX}\label{def}
We define the crowd index \Index\ as follows:
\vspace{-0.05in}
\begin{definition}[\Index]
A \Index\ $i^t=(\mathcal{P}_i^t,\mathcal{L}_i^t)$ is a pair that
represents an assignment of a set of workers in $\mathcal{U}$ to a
task $t$. Formally, it is described by a vector $\mathcal{P}_i^t$ of
length $m+2$, and a set of workers $\mathcal{L}_i^t$. $\mathcal{P}_i^t
= \langle v_t, q_{t_1}, \ldots, q_{t_m}, w_t \rangle$ contains the
value $v_t$ of task $t$, its expected minimum expertise $q_{t_i}$ for
each skill $s_i$, and its maximum cost $w_t$. $\mathcal{L}_i^t \subseteq
\mathcal{U}$ contains the workers assigned to index $i^t$.
\end{definition}

Consider Example~\ref{runningex} with $T=\{t_1,t_2,t_3\}$ for which three indexes are to be created offline. If workers $\{u_1,u_2,u_6\}$ are assigned to task $t_1$ with $W_1=W_2=0.5$, then the index for task $t_1$ will be, $i^{t_1} = (\langle 0.6, 0.74,0.58 \rangle, \{u_1,u_2,u_6\})$.

\eat{Although \Index\ are built offline using a given workload,
they are used to quickly select the appropriate workers, as a task 
arrives in real-time. A unique characteristic of \Index\ is the
necessity to adapt it as the workers become available or unavailable. Most
importantly, indexing should account for the third human factor,
namely acceptance ratio. 

Even when they are present in the system,
workers may decline a task. In what follows, we first describe the
{\em index design problem} followed by a discussion on how indexes are used
as tasks arrive, and how to perform dynamic maintenance of the
indexes.}
{\bf C-DEX Design Problem:}
We start with the KI-C problem described in Section \ref{dm}. We define $v_t$ to denote the {\em value} of each task $t$ in $T$ (in the beginning $v_t$ is $0$ for every task). The task value is associated with the current quality and cost of the task. More specifically, task value is calculated as a weighted linear combination of skills (higher is better) and cost (lower is better). The objective is to design an index \Index\ such that the sum of values $\mathcal{V} =\Sigma_{\forall t} v_t$ of all tasks it $T$ is maximized, while the problem constraints are satisfied.   

For a task $t$, its individual value $v_t$ and the global value $\mathcal{V}$  is defined in Equation~\ref{eqn:eq}.
\begin{equation} \label{eqn:eq}
    \text{Maximize } \mathcal{V}= \Sigma_{\forall t \in T} v_t 
\end{equation}
\vspace{-0.2in}
\begin{equation*}
    v_t=
    \begin{cases}
      W_1\times\Sigma_{\forall j \in \{1..m\}} q_{t_j}+W_2 \times (1-\frac{w_t}{W_t})
      &\text{if } q_{t_j} \geq Q_{t_j}\\
      &\wedge \ w_t\leq W_t \\
      0&\text{if } q_{t_j} < Q_{t_j}\\
      &\vee \ w_t>W_t
    \end{cases}
 \end{equation*}


where $W_1, W_2 \geq 0$ and $W_1+W_2=1$. \\
Note that the above formulation is a flexible incorporation of different skills and cost, letting the application select the respective weights, as appropriate.

Since \Index\ are {\em pre-computed for future use},  skills (or quality) and wages
are computed in an expected sense considering acceptance ratio, instead of the actual aggregates, as follows:
\vspace{-0.25in}
\begin{align*} 
\\ q_{t_j} = \Sigma_{\forall u \in \mathcal{U}} u_{t} \times p_u \times  u_{s_j} \geq Q_{t_j}, \forall j \in \{1..m\} 
\\ w_t = \Sigma_{\forall u \in \mathcal{U}} u_t \times p_u \times  w_u \leq W_t 
\\  u_{t} = [0/1]
\\  X_{l} \leq \Sigma_{\forall t \in T} \{u_{t}\} \leq X_{h}
\end{align*} 
The above formulation is to design an assignment of a user $u \in \mathcal{U}$ to a task $t \in T$ to generate the \Index\ set $\mathcal{I}$.


\vspace{-0.05in}
\subsubsection{Theoretical Analyses}\label{theo}
\vspace{-0.1in}
\begin{theorem}\label{nphard}
The \Index\ Design Problem is NP-Complete.
\end{theorem}
\begin{proof}
It is easy to see that the problem is in NP.  To prove
NP-completeness, we prove that the well known Multiple-Knapsack
Problem (MKP)~\cite{DBLP:books/fm/GareyJ79} is polynomial time
reducible to an instance of the \Index\ Design Problem, i.e., MKP
$\leq_{p}$ \Index\ Design Problem.

An instance of MKP is as follows: a pair $(\mathcal{B},\mathcal{S})$,
where $\mathcal{B}$ is a set of bins, and $\mathcal{S}$, a set of
items. Each bin $j \in \mathcal{B}$ has a capacity $c(j)$, and each
item $a$ has a size $s(a)$ and profit $p(a)$. The objective is to find a
subset $\mathcal{U} \subseteq \mathcal{S}$ of maximum profit such that
$\mathcal{U}$ has a feasible packing in $\mathcal{B}$. The decision
version of this problem is to find a feasible packing with
$\mathcal{U}$ using $|\mathcal{B}|$ bins, where $\mathcal{U} \subseteq
\mathcal{S}$, such that the total profit is $P$.

We reduce an instance of MKP to create an instance of the \Index\
Design Problem as follows. We assume that $\mathcal{S}= \{s\}$ (i.e., $m=1$),
$W_1=1, W_2=0$. The workload consists of $|\mathcal{B}|$ tasks (is equal to the number of indexes). Each bin $j$ represents an \Index\ 
$j$, and the capacity of the bin $c(j)$ represents the maximum
expected workers' wages assigned to index $j$ (i.e., $c_j=W_j$). In this simpler version
of the \Index\ Design problem, we assume that each \Index\ has the minimum quality
requirement of $0$, i.e., $q_{j} = 0$.

The item set $\mathcal{S}$ represents the worker set, where each item
$a$ is a worker $u$, $p(a)=u_s$, and $w_u = s(a)$, $p_u = 1$.
This creates the following instance of the \Index\ Design problem, where $v_j$ is
the value of $j$-th \Index, and $\mathcal{V}$ is the overall value:
\vspace{-0.4in}
\begin{align*}
\\ \mathcal{V}= \Sigma_{\forall j \in \mathcal{B}} v_j
\\  v_j= 1 \times q_{j} + 0 \times (1-\frac{w_j}{W_j}),
\\ q_{j} = \Sigma_{\forall u} u_{j} \times p_u \times  u_{s} \geq 0, 
\\ w_j = \Sigma_{\forall u} u_j \times p_u \times  w_u \leq W_j,
\\  u_{j} = [0/1], 0 \leq \Sigma_{\forall j \in \mathcal{B}} \{u_{j}\} \leq 1. 
\end{align*}
Given the above instance of the \Index\ Design Problem, the
objective is to create $|\mathcal{B}|$ \Index, such that
$\mathcal{V} = P$ and there exists a solution of the MKP problem with
total profit $P$, if and only if, a solution to our instance of the
\Index\ Design Problem exists.
\end{proof}

\noindent {\bf Effect of Constraints on C-DEX Design Problem:} Interesting theoretical properties of the optimization problem (Equation~\ref{eqn:eq}) are investigated under different conditions and constraints. In particular, we investigate the sub-modularity and monotonicity properties~\cite{submod} of the objective function that are heavily used in designing approximation algorithms in Section~\ref{approx}. 

\smallskip \noindent {\bf Submodular Function:} In general, if $\mathcal{A}$ is a set, a submodular function is a set function: $ f : 2^\mathcal{A} \to \mathbb{R}$ that satisfies the following condition:
For every $\mathcal{X},\mathcal{Y} \subseteq \mathcal{A}$ with $\mathcal{X} \subseteq \mathcal{Y}$ and every $x \in \mathcal{A} \backslash \mathcal{Y}$, we have $f(\mathcal{X} \cup \{x\}) - f(\mathcal{X}) \geq f(\mathcal{Y} \cup \{x\}) - f(\mathcal{Y})$.

\smallskip \noindent Value function $v_t$ for task $t$ satisfies this form: it maps each subset of the workers $\mathcal{S}$ from $\mathcal{U}$  to a real number $v_t$, denoting the value if that subset of workers are assigned to task $t$. Conversely, global optimization function $\mathcal{V}=\Sigma_{\forall t \in T} v_t$ is defined over a set of sets, each set maps an assignment of a subset of the workers from $\mathcal{U}$  to a task $t \in T$ with value $v_t$.

\smallskip \noindent {\bf Monotonic Function:} A real valued function $f$ defined on non-empty subsets of $\mathbb{R}$ is monotonic if $f( \mathcal{X} \cup \{x\}) \leq f( \mathcal{X} \cup \{y\}), \forall x \leq y, \mathcal{X} \subset \mathbb{R}$.
\begin{theorem}
The value function $v_t$ is not submodular in  the \Index\ Design problem, if $Q_{t_j}>0, \forall j \in \{1..m\}$ .
\end{theorem}
\vspace{-0.05in}
\begin{proof}
Sketch: 
Without loss of generality, we ignore the weights and the acceptance ratios of the workers for this proof. For the simplicity of exposition imagine $m=1$. Value $v_t$ remains $0$ until  $q_{t_1} \geq Q_{t_1}$. Consider a subset  $\mathcal{R} \subset \mathcal{S}$ and imagine $f'(\mathcal{R}) <  Q_{t_1}$, leading to $v_t = 0$. If an element $k$ is added to $\mathcal{R}$, if  $f'(\mathcal{R} \cup k) <  Q_{t_1}$, then $v_t = 0$. However, if  $f'(\mathcal{S}) \geq Q_{t_1}$, $v_t > 0$. Thus $v_t >0$, for $f'(\mathcal{S} \cup k)$. In such cases, it is easy to see,\\
$f'(\mathcal{S} \cup k) - f'(\mathcal{S}) > f'(\mathcal{R} \cup k) -  f'(\mathcal{R})$. This clearly violates the submodularity condition.  We omit the details for brevity.
\end{proof}
\vspace{-0.05in}
\begin{theorem}\label{locald}
The value function $v_t$  in the \Index\ Design problem is submodular but non-monotone, when $Q_{t_j} = 0, \forall j \in \{1..m\}$.
\end{theorem}
\vspace{-0.05in}
\begin{proof}
Sketch: As long as $Q_{t_j} = 0, \forall j \in \{1..m\}$, it could be proved that the increase in value by adding a worker $k$ to $\mathcal{S}$ is {\em strictly less} than adding $k$ to $\mathcal{R}$, where $\mathcal{R} \subset \mathcal{S}$, with the cost threshold $w_t \leq W_t$. Therefore, the following condition of submodularity i.e., ``diminishing return'' holds:
 $f(\mathcal{S} \cup k) -  f(\mathcal{S}) < f(\mathcal{R} \cup k) -  f(\mathcal{R})$.
At the same time, $v_t$ could increase or decrease when a worker is added (depending on whether the skill increase is more than the cost decrease or vice versa). Hence $v_t$ is non-monotone.
\end{proof}
\vspace{-0.1in}
\begin{lemma}\label{local}
The value function $v_t$  in \Index\ Design problem is submodular and monotonic, when $Q_{t_j} = 0, \forall j \in \{1..m\}$  and $W_2$=0.
\end{lemma}
\vspace{-0.1in}
\begin{theorem}\label{global}
The objective function $\mathcal{V}$ in the \Index\ Design problem is submodular and monotonic, when $Q_{t_j} = 0, \forall j \in \{1..m\}$ and $W_2=0$ and $X_l=0$.
\end{theorem}
\vspace{-0.05in}
\begin{proof}
Sketch: Consider our objective function $\mathcal{V} = \Sigma_{\forall t} v_t$  defined over a set of sets, where each set defines a subset of workers assigned to a task $t$ with value $v_t$. Adding a  worker $k$ to a set $\mathcal{R}$ (corresponds to task $t$) will impact  $v_t$ and therefore the overall $\mathcal{V}$. Without the skill threshold, i.e., $Q_{t_j} = 0, \forall j \in \{1..m\}$, if $k$ is added to $\mathcal{S}$ instead,  where $\mathcal{R} \subset \mathcal{S}$, the following condition of submodularity will hold: 
$f(\mathcal{S} \cup k) -  f(\mathcal{S}) < f(\mathcal{R} \cup k) -  f(\mathcal{R})$.
Furthermore, $\mathcal{V}$ strictly increases when $W_2=0$ and $X_l=0$ (i.e., a worker may not be assigned to any task) and ensures monotonicity. 
\end{proof}
\vspace{-0.1in}
\subsection{C-DEX$^+$}
\label{cdexp_problem}
Even though solved offline,  the computation time of \Index\ may still be very expensive, when the number of workers or tasks is large.  \vindex\ is a novel alternative towards that end, where the actual worker pool is intelligently {\em replaced} by a set of {\em Virtual Workers}, that are much smaller in count. \sys\ uses the Virtual Workers and the same workload to pre-compute a set of indexes, referred to as \vindex. \vindex\ enables efficient pre-computation, as well as faster assignments from workers to tasks.

 Intuitively, a Virtual Worker represents a set of ``indistinguishable'' actual workers, who are similar in skills and cost. For the simplicity of exposition, if we assume that in a given worker pool, there are $3$ workers who posses exactly same skill $s$  and cost $w$, then a single Virtual Worker $V$ could be created replacing those $3$ with skill $s$  and cost $w$. Obviously, when there are variations in the skills and costs of workers, the profile of $V$ needs to be defined conservatively - by taking maximum of individual worker's cost as $V$'s cost, and minimum  of individual worker's expertise, per skill. The formal definition of $V$ is:

\begin{definition}\label{def:def3}
Virtual Worker V :  V represents a set $n'$ of actual workers that are ``indistinguishable''. V is an $m+2$ dimensional vector,  $\langle V_{s'_1}, V_{s'_2}, \ldots, V_{s'_m}, V_{w'}, |n'| \rangle$ describing expected skill,expected wage, number of actual workers in $V$, where, 
$V_{s'_i}= \min_{\forall u \in n'} p_u \times u_{s_i}$,
$V_{w'}= \max_{\forall u \in n'} p_u \times w_u$.
\end{definition}
Consider Example~\ref{runningex} again, if $u_2$ and $u_5$ are grouped together to form a Virtual Worker $V$, then $V = \langle 0.21, 0.18, 2 \rangle$.

{\bf C-DEX$^+$ Design Problem:}
It is apparent that the Virtual Workers help reduce the size of the optimization problem. The formal definition of \vindex\ is:

\begin{definition}[\vindex]
A \vindex\ $i^{t_V}=(\mathcal{P}_i^{t_V},\mathcal{L}_i^{t_V})$ is a pair that
represents an assignment of a set of Virtual Workers in $\mathcal{N}$ to a
task $t$. $\mathcal{P}_i^{t_V},\mathcal{L}_i^{t_V}$ are similar to $\mathcal{P}_i^{t},\mathcal{L}_i^{t}$ and defined using the Virtual Worker set $\mathcal{N}$ .
\end{definition}

\subsection{Index Maintenance}\label{maintd}
A unique challenge that \sys\ faces is, even
if the most appropriate index is selected for a task, one or more
workers who were assigned to the task may not be available (for
example, they are not online or they decline the task). Note that the acceptance ratio only quantifies an overall availability of a worker, but not for a particular task. 
Therefore, \sys\ needs to dynamically find a replacement for unavailable workers. At the same time, \sys\ needs to strictly ensure non-preemption of the workers, since workers who accepted a task are required to continue their work. \sys\ proposes several principled solutions that make use of the theoretical analysis in Section~\ref{theo}.

Furthermore, \sys\ has to deal with scenarios where, new workers could subscribe to the system any time, or some existing ones could delete their accounts. Similarly, as
existing workers complete more tasks, the system may update their
profile (refine their skills for example). How to learn the profile of
a new worker or an updated profile of an existing worker is
orthogonal to this work. What we are interested in here is {\em how
\sys\ makes use of those updates, by maintaining them
incrementally.}

We therefore investigate principled solution towards incremental index maintenance for four scenarios: (1) worker replacement,(2) worker addition, (3) worker deletion, (4) worker profile update.

\vspace{-0.05in}
\section{Optimal Algorithms}\label{alg1}
Section~\ref{precom} proposes the \Index\ building solution, whereas, Section \ref{maintenance} discusses the maintenance. 

\subsection{C-DEX Design (offline phase)}\label{precom}
Recall Theorem~\ref{nphard} and note that the \Index\ Design Problem is proved to be NP-hard. \sys\ proposes an integer linear programming (ILP) based solution that solves the optimization problem defined in Equation~\ref{eqn:eq} optimally satisfying the constraints.


 While the optimization problem is a linear combination of weights and skills, unfortunately, the decision variables (i.e., $u_t$'s) are required to be
integers. More specifically, \Index\ set are created by generating a
total of $n \times |T|$ boolean decision variables, and the solution of this
optimization problem assigns either a $1/0$ to each variable, denoting
that a worker is assigned to a particular task, or not. These
integrality constraints make the above formulation an Integer Linear
Programming (ILP) problem~\cite{DBLP:conf/ipco/2013}. A solution to the ILP problem would
perform an assignment of a worker to a task in $T$. Once the
optimization problem is solved, an index $i^{t}$ is designed for each
task in the workload, $\langle \mathcal{P}_i^t,\mathcal{L}_i^t \rangle$ is
calculated. Algorithm~\ref{alg:optrange} summarizes the pseudocode. 

Given Example~\ref{runningex}, when $W_1=W_2=0.5$, the best allocation gives rise to $\mathcal{V} = 1.98$, with the following worker to task allocation: $u_1= \{t_1\}, u_2= \{t_1,t_2\}, u_3=\{t_3\},u_4=\{t_2,t_3\},u_5=\{t_2,t_3\}, u_6=\{t_1,t_3\}$. This creates the following $3$ indexes:\\ $i^{t_1} = (\langle 0.6, 0.74,0.58 \rangle, \{u_1,u_2,u_6\})$, \\ $i^{t_2} = (\langle 0.59, 0.75,0.71 \rangle, \{u_2,u_4,u_5\})$, \\ and $i^{t_3} = (\langle 0.79, 1.15, 1.13 \rangle\,\{u_3,u_4,u_5,u_6\})$.

\begin{algorithm}[t]
\caption{Optimal \Index\ Design Algorithm}
\label{alg:optrange}
\begin{algorithmic}[1]
\begin{small}
\REQUIRE   Workload $T$ \\
\STATE Solve the \Index\ Design ILP to get an assignment of the $u_t=0/1$, where $u$ is a worker, and $t \in T$.
\STATE using $u_t$, for each $t \in T$, compute and output $i^{t} = \langle \mathcal{P}_i^t,\mathcal{L}_i^t \rangle$
\RETURN Index set $\mathcal{I}$
\end{small}
\end{algorithmic}
\end{algorithm}

Unfortunately, ILP is also NP-Complete~\cite{DBLP:books/fm/GareyJ79}. The commercial implementations of ILP use techniques such as Branch and Bound~\cite{DBLP:conf/ipco/2013} with the objective to speed up the computations. Yet, computation time is mostly non-linear to the number of associated variables and could become exponential at the worst case.

\subsection{C-DEX Maintenance (online phase)}\label{maintenance}
We design index maintenance algorithms, which generate optimal solutions under the non-preemption constraint (constraint no.3, Section \ref{constraints}).
Non-preemption of workers enforces that the existing assignment of an available worker can not be disrupted, only new assignments could be made if she is not maxed-out. Under this assumption, all four incremental maintenance strategies described below are optimal.

\subsubsection{Replacing Workers} \label{rep}
To dynamically find a replacement for unavailable workers, without disrupting already made assignments, we formulate a {\em marginal ILP} and solve the problem optimally only with the available set of workers.

We illustrate the scenario with an example. Suppose that after
the most appropriate index $i^{t}$ is selected for task $t = \langle
Q_{t_1}, Q_{t_2}, \ldots, Q_{t_m}, W_t \rangle$ using
Equation~\ref{eqn:eq}, a subset of workers in $\mathcal{L}_i^t$ is
unavailable/ or declines to work on $t$. Imagine that the 
quality of $i^{t}$ declines to $q'_{t_j}$ from $q_{t_j}$, for skill $j$, $\forall j \in
m$, and the cost declines to $w'_t$ from $w_t$, as some workers do not accept the task. Consequently, the value of $i^{t}$ also declines, let us say, to
$v'_t$ from $v_t$.  From the worker pool $\mathcal{U}$, let us imagine
that a subset of workers $\mathcal{U}'$ are available and their
current assignment has not maxed out (i.e., $\mathcal{C}_{u'} < X_h$).
To find the replacement of the unavailable workers, \sys\ works as
follows: It formulates a marginal ILP problem
with the same optimization objective for $t$, only with the workers in
$\mathcal{U}'$. More formally, the task is formulated as:
\vspace{-0.1in}
\begin{equation} \label{eqn:eq5}
\text{Maximize }  v''_t  = v'_t + W_1 \times \Sigma_{\forall j \in m}q''_{t_j} + W_2 \times (1-\frac{w''_t}{W_t})
\end{equation}
\vspace{-0.4in}
\begin{align*}
\\  W_1+W_2 = 1,
\\ q''_{t_j} = \{ q'_{t_j} + \Sigma_{\forall u' \in \mathcal{U'}} u'_{j} \times p_{u'} \times  u_{s_j} \} \geq Q_{t_j}, 
\\ w''_t = \{ w'_t + \Sigma_{\forall u' \in \mathcal{U'}} u'_t \times p_{u'} \times  w_{u'} \} \leq W_t,   u'_{t} = [0/1].
\end{align*}

\begin{lemma}
The marginal ILP in Equation~\ref{eqn:eq5} involves only $|\mathcal{U}'|$ variables.
\end{lemma}

The above optimization problem is formulated only for a task $t$ and
considering only $|\mathcal{U}'| < < |\mathcal{U}|$ workers. It is
incremental in nature, as it ``builds'' on the current solution
(notice that it uses the declined cost, skills, and value in the
formulation), involving a much smaller number of variables and leading to
small latency. Moreover, this strategy is fully aligned with the
optimization objective that \sys\ proposes. After this formulation is
solved, $\mathcal{L}_i^t$ is updated with the new workers for which the
above formulation has produced $u'_t = 1$.


\subsubsection{Adding New Workers}
Assume that a set $\mathcal{A}$ of new workers has subscribed to the
platform. The task for \sys\ is to decide whether (or not) to assign
those workers to any task in $T$, and if yes, what should be
the assignment. Note that \sys\ already has assigned the existing worker set
$\mathcal{U}$ to the tasks in $T$ and they can not be preempted.

\eat{
\begin{lemma} \label{l1}
The \Index\ Design problem in Section~\ref{index} is monotonic, when $X_l = 0$.
\end{lemma}

\begin{proof}
(sketch): We prove the lemma by contradiction. Imagine that a new
  worker $u$ has been assigned to a task in $T$, that
  resulted in a decrease to the global value $\mathcal{V}$. This could
  happen because the optimization function is a combination of quality
  and cost. Adding a worker always improves the quality, but also
  increases the cost, forcing the cost factor of the optimization
  function to decline. Imagine that the new value be $\mathcal{V}' <
  \mathcal{V}$. But since $X_l=0$ (i.e., a new worker may not be
  assigned to any task at all), this is a contradiction, because in
  that case, $u$ will not be assigned to any task in
  $T$. Hence the contradiction.
\end{proof}
}

The overall idea is to solve optimally a marginal ILP only with the new workers in $\mathcal{A}$ and tasks $T$, without making any modifications to the existing
assignments of the $\mathcal{U}$ workers to the $T$
tasks. Formally, the problem is formulated as follows:
\vspace{-0.05in}
\begin{equation}\label{eqn:eq6}
\text{Maximize }  \Sigma_{\forall t \in T} 
\end{equation}
\vspace{-0.4in}
\begin{align*}
\\ \{ v'_t \} = v'_t = v_t + W_1 \times \Sigma_{\forall j \in m}q'_{t_i} + W_2 \times (1-\frac{w'_t}{W_t})
 \\  W_1+W_2 = 1,
\\ q'_{t_j} = \{ q_{t_j} + \Sigma_{\forall u \in \mathcal{A}} u_{t} \times p_u \times  u_{s_j} \} \geq Q_{t_j}, 
\\ w'_t = \{ w_t + \Sigma_{\forall u \in \mathcal{A}} u_t \times p_u \times  w_u \} \leq W_t 
\\  u_{t} = [0/1], 0 \leq \Sigma_{\forall t \in T} \{u_{t} \in \mathcal{A}\} \leq X_h.
\end{align*}


\begin{lemma}
The optimization problem in Equation~\ref{eqn:eq6} involves only $|\mathcal{A}| \times |T|$ variables.
\end{lemma}


\subsubsection{Deleting Workers}
In principle, the treatment of worker deletion is analogous to that of
worker replacement strategies in Section~\ref{rep}.
Basically, the idea is to
determine the decreased quality, cost, and value of each of the tasks that are
impacted by the deletion, and then re-formulate an optimization
problem only with those tasks, and the remaining workers who are not
maxed-out yet (i.e., $\mathcal{C}_u < X_h$) on their assignment, using the current
quality, cost, and value. Similar to Section~\ref{rep}, this
formulation is also a marginal ILP that is incremental in nature, and involves a smaller number of variables.  We omit further discussion on this for brevity.

\subsubsection{Updating Worker Profiles}
Interestingly, the handling of updates in worker profile is also
incremental in \sys. If the skill, wage, or acceptance-ratio of a
subset $\mathcal{A'}$ of workers gets updated, \sys\ first updates the
respective {\em value} of the tasks (where these workers were
assigned), by discounting the contribution of the workers in
$\mathcal{A'}$. After that, a smaller optimization problem is
formulated involving only $\mathcal{A'}$ workers and $T$
tasks.  After discounting the contribution of the workers in
$\mathcal{A'}$, if the latest value of a task $t$ is
$v'_t$\footnote{\small If none of the workers in $\mathcal{A'}$
  contributed to $t$, then $v'_t=v_t$.}, current quality on skill $j$
is $q'_{t_j}$, and current cost is $w'_t$, then the optimization
problem is formulated as:
\vspace{-0.1in}
\begin{equation}\label{eqn:eq7}
\text{Maximize }  \Sigma_{\forall t \in T} \{ v'_t \}
\end{equation}
where, 
\vspace{-0.2in}
\begin{align*}
v''_t = v'_t + W_1 \times \Sigma_{\forall j \in \{1..m\}}q''_{t_i} + W_2 \times (1-\frac{w''_t}{W_t}), 
\\  W_1+W_2 = 1,
\\ q''_{t_j} = \{ q'_{t_j} + \Sigma_{\forall u \in \mathcal{A'}} u_{t} \times p_u \times  u_{s_j} \} \geq Q_{t_j}, 
\\ w''_t = \{ w'_t + \Sigma_{\forall u \in \mathcal{A'}} u_t \times p_u \times  w_u \} \leq W_t 
\\  u_{t} = [0/1], X_l \leq \Sigma_{\forall t \in T} \{u_{t} \in \mathcal{A'}\} \leq X_h
\end{align*}

Similar to the previous cases, the proposed solution is principled and
well-aligned with the optimization objective that \sys\ proposes. The
solution involves only $|\mathcal{A'}| \times |T|$ variables,
and our experimental study corroborates that it generates the output
within reasonable latency.

\section{Approximation Algorithms}\label{approx}
The optimal algorithm presented in Section~\ref{alg1} may be very expensive during index building as well as maintenance time, since the ILP-based solution may have exponential computation time at the worst case. To expedite both of these steps,  two approximate solutions  are discussed next: a) A greedy approximate solution for \Index\ that has provable approximation factor under certain conditions. b) A clustering-based solution \vindex\ which offers high efficiency but may give approximate result.

\subsection{Greedy Approximation for C-DEX} \label{greedycdex}
Next we describe the greedy strategies for \Index\ creation and adaptive maintenance, both guaranteed to run in polynomial time. The quality of the results is approximate but the approximation factor can be guaranteed under certain conditions.

\subsubsection{Approximate C-DEX Design (offline phase)} \label{offg}
The approximate \Index\ design algorithm 
{\tt Offline-CDEX-Approx} 
follows a greedy  strategy for index building which admits a provable approximation factors under certain conditions. 
Given the pool of tasks and workers, it iteratively adds a worker to a task such that the addition ensures the {\em highest marginal gain} in $\mathcal{V}$ in that iteration, while ensuring the quality, cost, and tasks-per-worker constraints. Imagine a particular instance of  {\tt Offline-CDEX-Approx} on Example~\ref{runningex} after first iteration. After a single worker assignment (first iteration will assign one worker to one of the indexes), if only $u_1$ is assigned to $i^{t_1}$ and nobody  to $i^{t_2}$ and $i^{t_3}$ yet, then the algorithm may select $u_6$ to assign to $i^{t_3}$  in the second iteration to ensure the highest marginal gain in $\mathcal{V}$.
\begin{theorem}
{\tt Offline-CDEX-Approx} has an approximation factor of $(1-1/e)$,  when $Q_{t_j} = 0, \forall j \in \{1..m\}$ and $W_2=0$ and $X_l=0$.
\end{theorem}

\begin{proof}
Sketch: The proof relies on our theoretical analyses in Section~\ref{theo} and on the fact that the optimization function $\mathcal{V}$ becomes submodular and monotonic under the above-mentioned conditions. We omit the details for brevity.
\end{proof}

\begin{lemma}
The run time of algorithm {\tt Offline-CDEX-Approx} is polynomial, i.e., $O(X_h \times \ |\mathcal{U}| \times |T|)$.
\end{lemma}

\subsubsection{Approximate C-DEX Maintenance (online phase)}\label{greedymain}
We discuss four greedy maintenance strategies next that are incremental and designed ensuring worker non-preemption.
{\bf Replacing Workers:} 
After a task arrives if one or more of the assigned workers to this task are not available, an efficient greedy solutions is proposed by selecting replacement workers from the available pool. This strategy leads to a provable approximation algorithm, when $Q_{t_j} = 0, \forall j \in \{1..m\}$ and $W_2=0$. We describe the greedy algorithm  
{\tt Online-CDEX-Approx}
 next.

Given a set of unavailable workers in $\mathcal{L}_i^t$, \sys\ performs a simple iterative  greedy replacement from the available pool of workers $\mathcal{U}'$. In a given iteration, the idea is to select that worker from the available pool and add her to $\mathcal{L}_i^t$ which results in the {\em highest marginal gain in $v_t$}. This iterative process continues until the cost constraint exceeds. 
This greedy algorithm is approximate in nature but admits a provable approximation factor under certain conditions.

\begin{theorem}
Algorithm {\tt Online-CDEX-Approx} admits an approximation factor of $1-1/e$, when $Q_{t_j} = 0, \forall j \in \{1..m\}$ and $W_2=0$.
\end{theorem}

\begin{proof}
Sketch: We omit the details for brevity; however, our proof uses the monotonicity and submodularity property of $v_t$ as proved in Section~\ref{theo} under these conditions.
\end{proof}

Of course, unless the above conditions are satisfied, the above approximation factor does not theoretically hold.

 \begin{lemma}
 The run time of 
{\tt Online-CDEX-Approx}
 is polynomial, i.e., $O(|\mathcal{U}'|)$.
 \end{lemma}
 
{\bf Addition of New Workers:} Our proposed greedy solution is similar in principle to the offline greedy approximation algorithm described in Section~\ref{offg}. New workers are to be assigned to the pre-computed indexes based on the highest marginal gain in value without disrupting the existing allocation of the current workers. In order to satisfy any theoretical guarantee, this objective function has to relax quality and cost threshold, number of tasks per worker, and make $W_2=0$. We omit further discussions for brevity. 
 
{\bf Deletion of Workers:}
This solution is akin to that of the greedy worker replacement strategy described above. It admits the exact same set of theoretical claims under similar conditions as described above. 

{\bf Updates of Worker Profile:} 
If the skill, wage, or acceptance ratio of a
subset $\mathcal{A'}$ of workers gets updated, \sys\ first updates the
respective {\em value} of the tasks (where these workers were
assigned), by discounting the contribution of the workers in
$\mathcal{A'}$. After that, it adapts the {\tt Offline-CDEX-Approx} (Section~\ref{offg}) involving $\mathcal{A'}$ workers and $T$
tasks. It iteratively adds a worker in $\mathcal{A'}$ to a task in $T$ based on the highest marginal gain in value, as well as satisfy the skill, cost, and number of workers per task constraint. Akin to {\tt Offline-CDEX-Approx}, this algorithm does not satisfy the $(1-1/e)$ approximation factor, unless $Q_{t_j} = 0, \forall j \in \{1..m\}$ and $W_2=0$ and $X_l=0$.

\subsection{C-DEX$^{+}$}\label{alg2}
Next, we present our second approximate solution \vindex\ for index building and adaptive maintenance based on clustering of workers. This solution is approximate yet very efficient, since it replaces the actual set of workers with a very small set of Virtual Workers (a VIrtual Worker represents a set of ``indistinguishable'' actual workers, who are similar in skills and cost, as defined in section \ref{cdexp_problem}).

\subsubsection{C-DEX$^+$ Design (offline phase)}
We work in two steps: 1) Creating virtual Workers and 2) Designing the \vindex. 

{\bf Creating Virtual Workers}
First,  a set $\mathcal{N}$ of Virtual Workers is created, given $\mathcal{U}$. Intuitively, a Virtual  Worker $V$ should represent a set of workers who are {\em similar} in their profile. In \sys, Virtual Workers are created by performing multi-dimensional clustering~\cite{DBLP:books/mk/HanK2000} on $\mathcal{U}$, and considering a {\em threshold} $\alpha$ that dictates the maximum {\em distance} between any worker-pairs inside the same cluster. The size of the Virtual Worker set  $\mathcal{N}$ clearly depends on 
$\alpha$, a large value of $\alpha$ leads to  smaller $|\mathcal{N}|$, and vice versa. Interestingly, this allows flexible design, as the appropriate trade-off between the quality and the cost could be chosen by the system, as needed. 
Formally, given $\mathcal{U}$ and $\alpha$, the task is to design a set of Virtual Workers, such that the following condition is satisfied:
\begin{equation*}
\forall u, u': u \in V, u' \in V, \mathit{Dist}(u,u') \leq \alpha
\end{equation*}
Our implementation uses a variant of Connectivity based Clustering~\cite{DBLP:books/mk/HanK2000} considering Euclidean distance to that end.

For example, if $\alpha = 0.25$, Example~\ref{runningex} will create $|\mathcal{N}| =2$ Virtual Workers; $V_1$ with $\{u_1,u_2,u_3,u_5\}$ and $V_2$ with$ \{u_4,u_6\}$; $V_1 = \langle 0.08, 0.18, 4 \rangle$  and $V_2 = \langle 0.3, 0.36, 2 \rangle$.

\eat{
It is apparent that the Virtual Workers help reducing the size of the optimization problem. The formal definition of the index design problem is defined next:

\begin{definition}[\vindex]
A \vindex\ $i^{t_V}=(\mathcal{P}_i^{t_V},\mathcal{L}_i^{t_V})$ is a pair that
represents an assignment of a set of Virtual Workers in $\mathcal{N}$ to a
task $t$. $\mathcal{P}_i^{t_V},\mathcal{L}_i^{t_V}$ is defined similar to $\mathcal{P}_i^{t},\mathcal{L}_i^{t}$ respectively, but considering only Virtual Worker set $\mathcal{N}$ .
\end{definition}

Intuitively the index design problem works in two steps:
}

{\bf Designing C-DEX$^+$: } 
For a Virtual Worker $V$ with $|n'|$ actual workers, a counter $\mathcal{C}_V$ is created stating the maximum assignments of $V$, i.e.,  $\mathcal{C}_V = |n'| \times X_h$. An ILP is designed analogous to Section~\ref{alg1} with $|\mathcal{N}|$ workers, and all the tasks in $T$. Additionally, a total of $2|\mathcal{N}|$ constraints are added; one per $V$, stating that the maximum and the minimum allocation of $V$ are $\mathcal{C}_V$ and $(|n'| \times X_l$), respectively. 
\begin{lemma}
The optimization problem for  C-DEX$^+$ involves only $|\mathcal{N}| \times |T|$ variables
\end{lemma}
\vspace{-0.05in}
Using the above lemma, it is easy to see that the ILP is likely to get solved faster for \vindex, as it involves less number of variables.

Example~\ref{runningex} gives rise to $2$ virtual workers $V_1,V_2$ when $\alpha = 0.25$. Two additional maximum allocation constraints will be added to the optimization problem, such that $\mathcal{C}_{V_1}=4,\mathcal{C}_{V_2}=2$. Therefore, the index-design problem with Virtual Workers could be solved for $3$ tasks and $2$ Virtual Workers, involving only $3 \times 2 = 6$ decision variables, instead of $6 \times 3 = 18$ variables that \Index\ has to deal with. While this solution is much more efficient compared to \Index, it may give rise to approximation to the achieved quality (i.e., in the objective function value $\mathcal{V}$), as the search space for the optimization problem gets further restricted with the Virtual Workers, leading to sub-optimal solution for $\mathcal{V}$. Interestingly, our empirical results shows that this alternative solution is efficient, yet the decline in the overall quality is negligible. 

The output of the optimization problem is the set of task indexes $\mathcal{I}_V$ using virtual workers. Considering Example~\ref{runningex}, $\mathcal{I}_V = \{i^{t_{V_1}},i^{t_{V_2}},i^{t_{V_3}}\}$. For task $t_1$, created \\ $i^{t_{V_1}} = (\langle 0.38, 0.76, 1.08 \rangle, \{V_1,V_1,V_2,V_2\})$, when $W_1=W_2=0.5$. The individual worker to task assignment could be performed after that by a simple post-processing. 

\subsubsection{C-DEX$^+$ Maintenance (online phase)}\label{maintenancev}
Recall that the maintenance strategies are designed for $4$ different scenarios, enforcing worker non-preemption constraint.

{\bf Replacing Workers:} \vindex\ designs a marginal ILP involving task $t$, and all the Virtual Workers whose current  $\mathcal{C}_V > 0$, akin to its \Index\ counterpart. Once the solution is achieved, individual worker assignment could be performed with a post-processing algorithm, in a round robin fashion, by keeping track of individual $L_V$'s. 

{\bf Addition of New Workers:} First, the existing Virtual Worker set $\mathcal{N}$  needs to get updated. Interestingly, since $\alpha$ is pre-determined, the new workers could be accommodated with incremental clustering, just by forming new clusters (i.e., creating new Virtual Workers) involving those additions, without having to re-create the entire $\mathcal{N}$ from scratch. After that, a smaller ILP is formulated only involving the Virtual Workers that are affected by the updates, considering existing partial assignments, akin to Equation~\ref{eqn:eq6}. We omit the details for brevity. 

{\bf Deletion of Workers:} The handling of worker deletion is akin to addition, in the sense, first \sys\ propagates these updates incrementally to the Virtual Worker set $\mathcal{N}$. To satisfy the pre-defined $\alpha$, it accounts for those remaining actual workers from each of the Virtual Worker $V$,  that has atleast one deleted worker. It reruns a smaller clustering solutions only involving those remaining workers. After $\mathcal{N}$ gets updated, the rest of the maintenance is exactly same as what is discussed in handling deletion inside Section~\ref{maintenance}. We omit the details for brevity. 

{\bf Updates of Worker Profile:} Similarly, if \sys\ gets to have updated profile of some of the workers, it first updates the Virtual Workers set by solving a smaller clustering problem, akin to deletion. With the updated Virtual Workers set, the rest of the maintenance here is same as solving a marginal ILP involving only the  updated Virtual Workers, as has been discussed in Section~\ref{maintenance} for maintaining profile updates.

\vspace{-0.1in}
\section{Experimental Evaluation}\label{exp}
\noindent We perform 2 different types of experiments: i) Real data experiments - conducted involving $250$ AMT~\footnote{\small Amazon Mechanical Turk, www.mturk.com} (AMT) workers in 3-different stages; ii) Synthetic data experiments - conducted using an event-based crowd simulator. The real-data experiments aim at evaluating the proposed approach in terms of  quality and feasibility, while the synthetic ones aim at validating its scalability and quality.

\subsection{Real Data Experiments}
The purpose of these experiments is to evaluate our approach in terms of feasibility and quality. We study {\em feasibility} since the current paid crowdsourcing platforms (like AMT) do not support KI-C task development and thus this is one of the first studies trying to optimize KI-C task production in such an environment. We study {\em quality} with the aim to measure the key qualitative axes of the knowledge produced by the hired workers. 

Overall the study is designed as an application of collaborative document writing by AMT workers selected using \sys. These results are compared to the respective results achieved using 2 representative rival strategies: {\tt Benchmark} (workers self-appoint themselves to articles after a skill-based pre-selection process, akin to how the current paid platforms work) and {\tt Online-Greedy} (workers are assigned to the available tasks taking into account the workers' marginal utility on each task; this is the adaptation of one of the latest state-of-the-art online task assignment algorithms~\cite{chienJuHo}. Workers are asked to produce documents on 5 different topics (KI-C tasks) of current interest: 1)Political unrest in Egypt, 2) NSA document leakage, 3)Playstation (PS) games, 4) All electric cars and 5) Global Warming. For simplicity and ease of quantification, we consider that each task requires one skill (i.e. expertise on that topic). The user study is conducted in 3 stages.

\vspace{-0.05in}
\subsubsection{Stage 1 - Worker Profiling}
In this stage, we hire $20$ AMT workers per task, totaling $100$ unique workers. The workers are informed that a subset of them will be invited (through email) in Stage 2 to collaboratively write a document on that topic. 
We design a set of $8$  multiple choice questions per task, assessing the workers' knowledge over facts related to the task (e.g., on Egypt -  \emph{``What is the name of the place in Cairo where the protests took place?}'' with possible answers: Tahrir Square, Mubarak Plaza, Al Azhar Square, or on the NSA leakage topic: \emph{``Who is Adrian Lemo?} with possible answers: A computer hacker, A federal agent, Both). The skill of a worker is then calculated as the percentage of her correct answers. Workers are also asked questions to extract their acceptance ratio and wage. Figure~\ref{fig:phase1} shows the quantification of worker profile distributions for the ``Egypt'' task. Worker profiles for the other topics exhibit similar distributions and are omitted for brevity. A strong positive correlation among workers' skill and their wage is also observed.

\eat{
\begin{figure*}
\begin{minipage}[t]{0.22\textwidth}
    \includegraphics[height=30mm, width = 60mm]{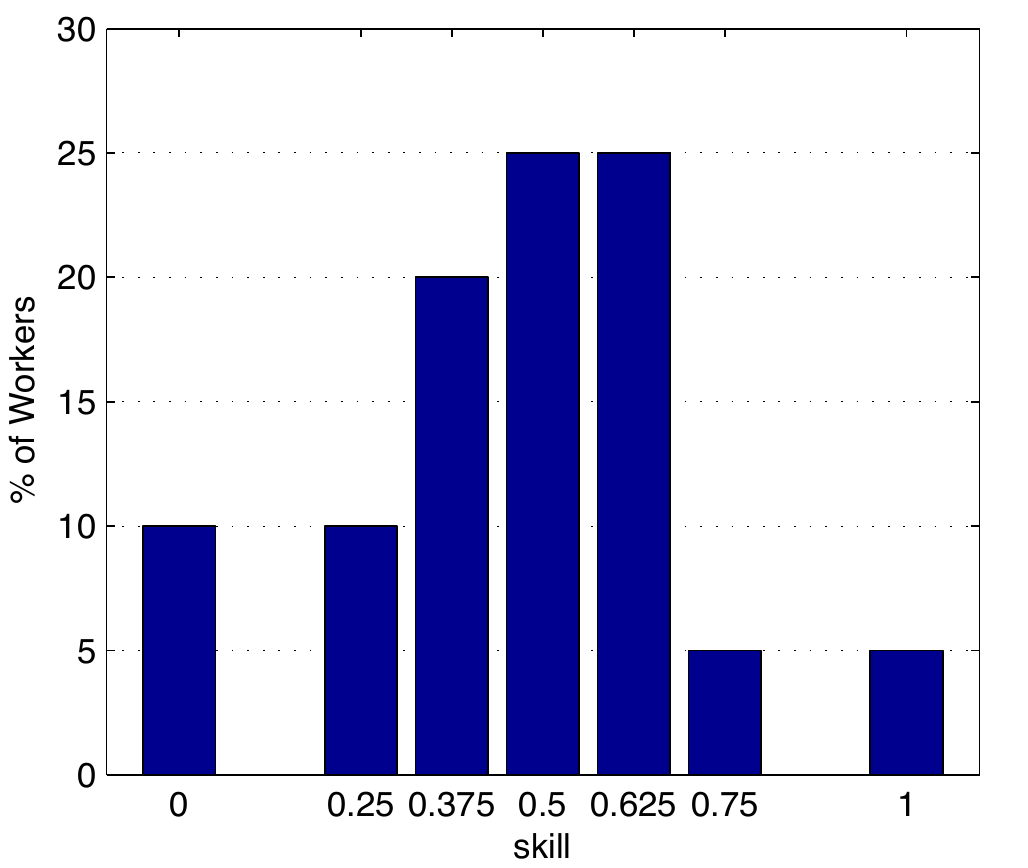}
    \caption{\small{Skill distribution}}
    \label{fig:scn1}
\end{minipage}
\hspace{5mm}
\begin{minipage}[t]{0.22\textwidth}
    \includegraphics[height=30mm, width = 60mm]{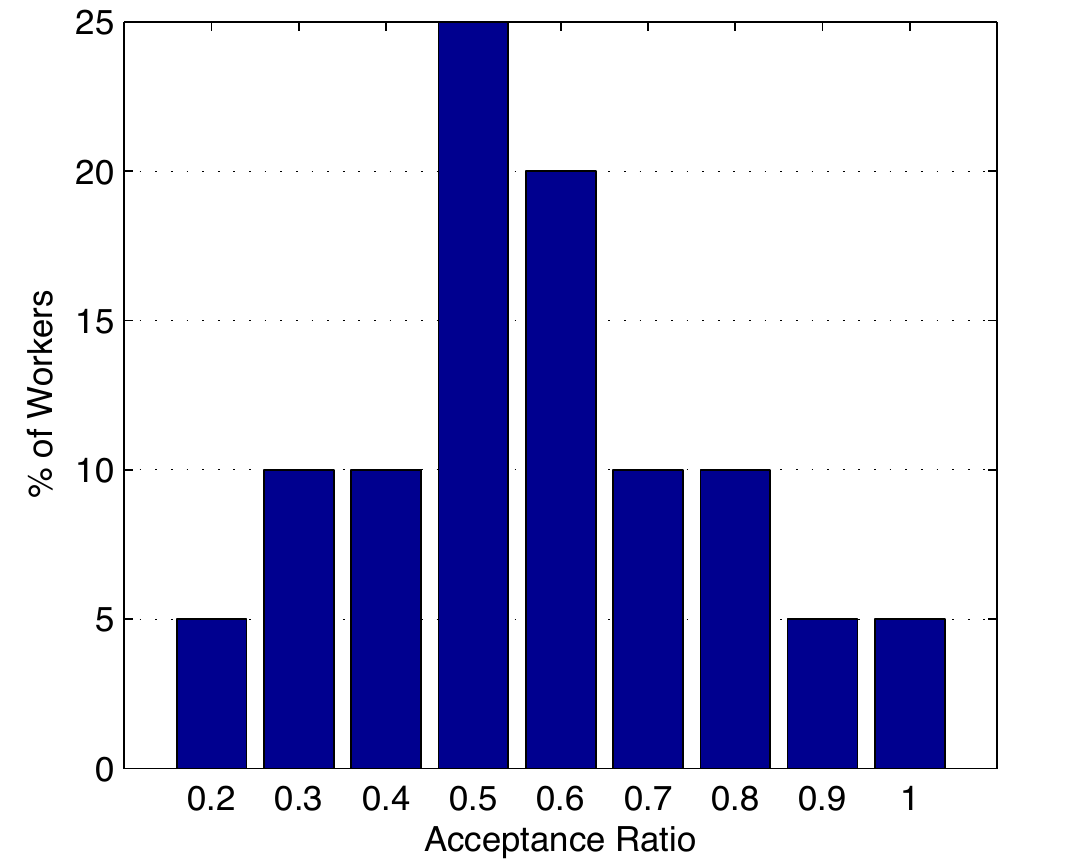}
    \caption{\small{Acceptance ratio distribution}}
    \label{fig:scn2}
\end{minipage}
\hspace{5mm}
\begin{minipage}[t]{0.22\textwidth}
    \includegraphics[height=30mm, width = 60mm]{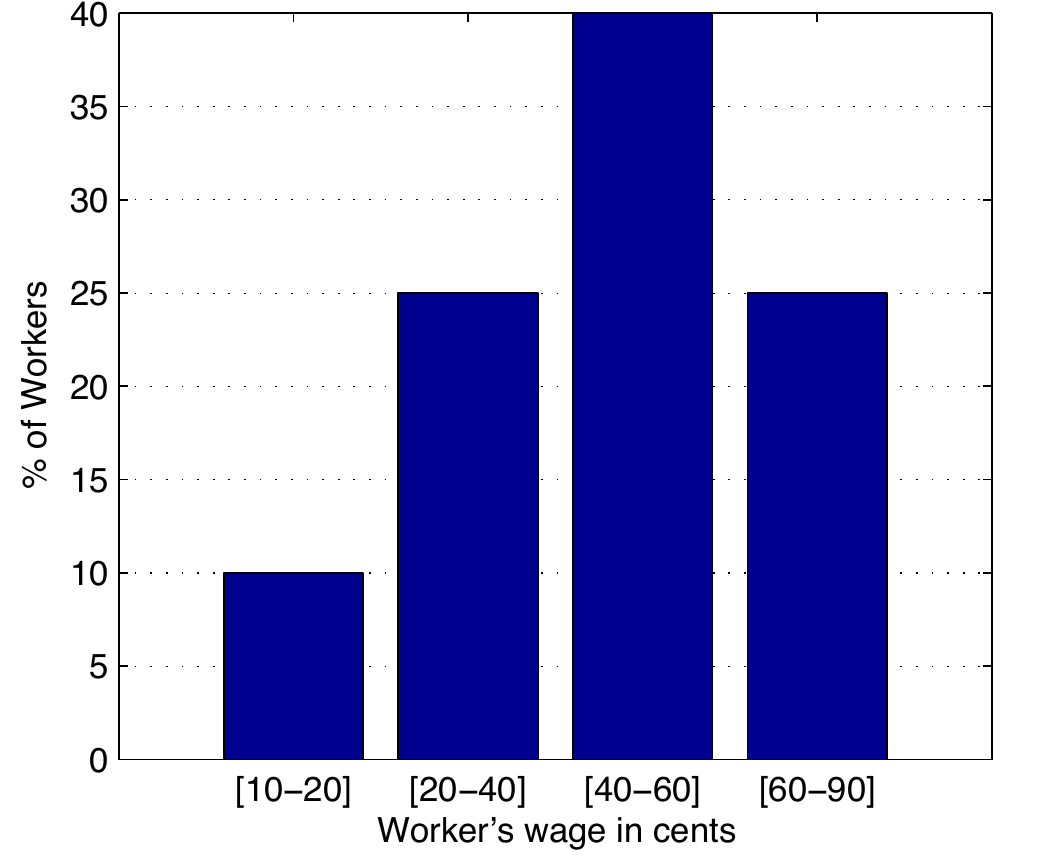}
   \caption{\small{Worker wage distribution}}
    \label{fig:scn3}
\end{minipage}
\hspace{5mm}
\begin{minipage}[t]{0.22\textwidth}
    \includegraphics[height=30mm, width = 50mm]{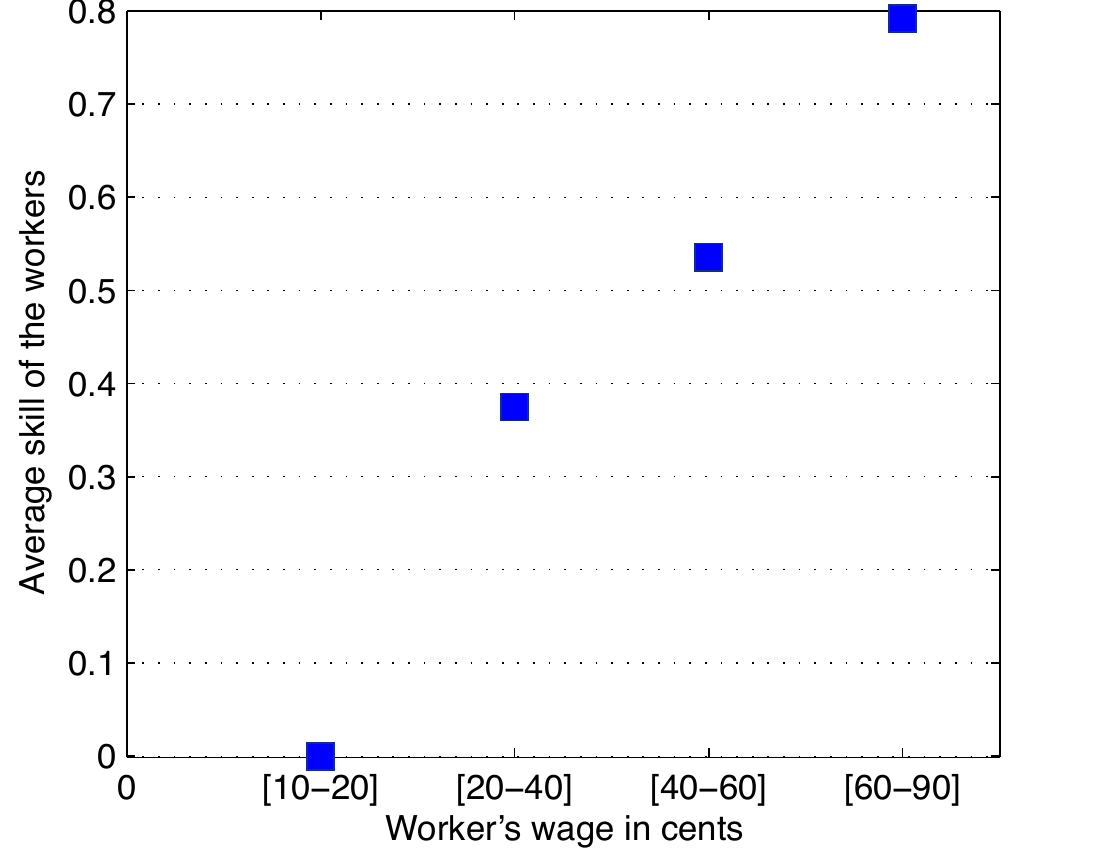}
   \caption{\small{Strong positive correlation between worker skill and wage}}
    \label{fig:scn4}
\end{minipage}
\end{figure*}
}

\begin{figure*}[t]
        \centering
        \begin{subfigure}[a]{0.23\textwidth}
                \includegraphics[width=\textwidth]{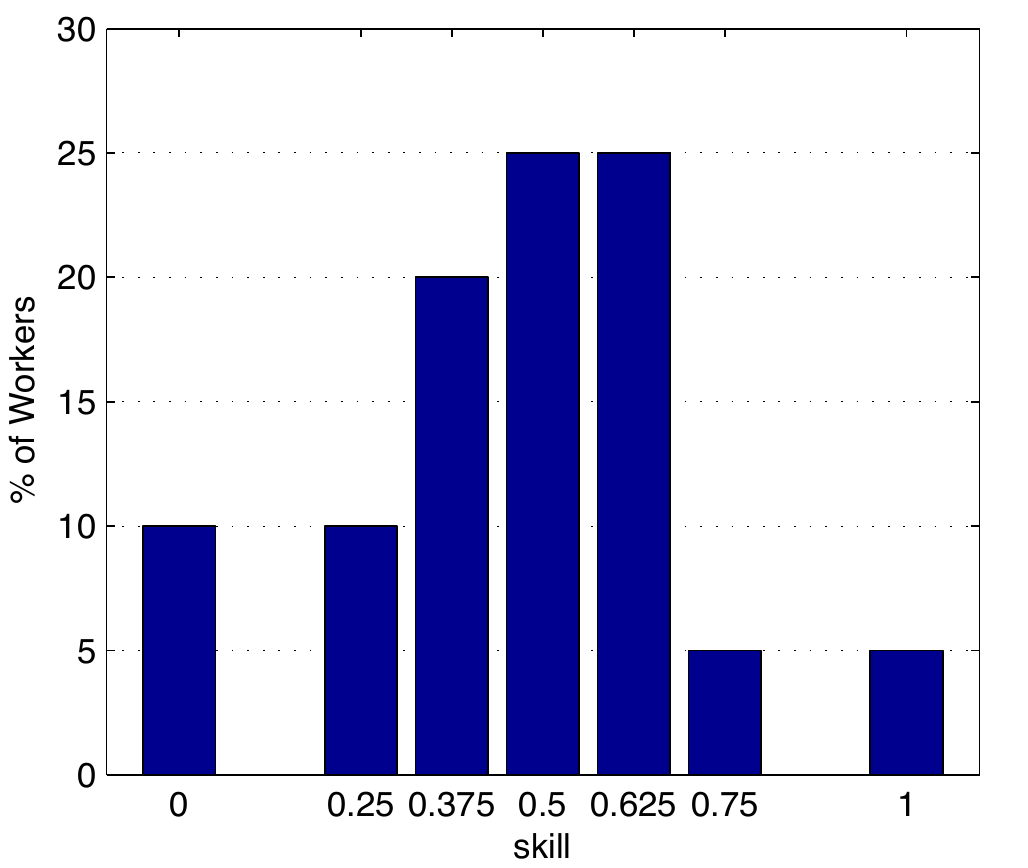}
\vspace{-0.05in}                
                \caption{Skill distribution}
                \label{fig:scn1}
        \end{subfigure}
        \begin{subfigure}[a]{0.23\textwidth}
                \includegraphics[width=\textwidth]{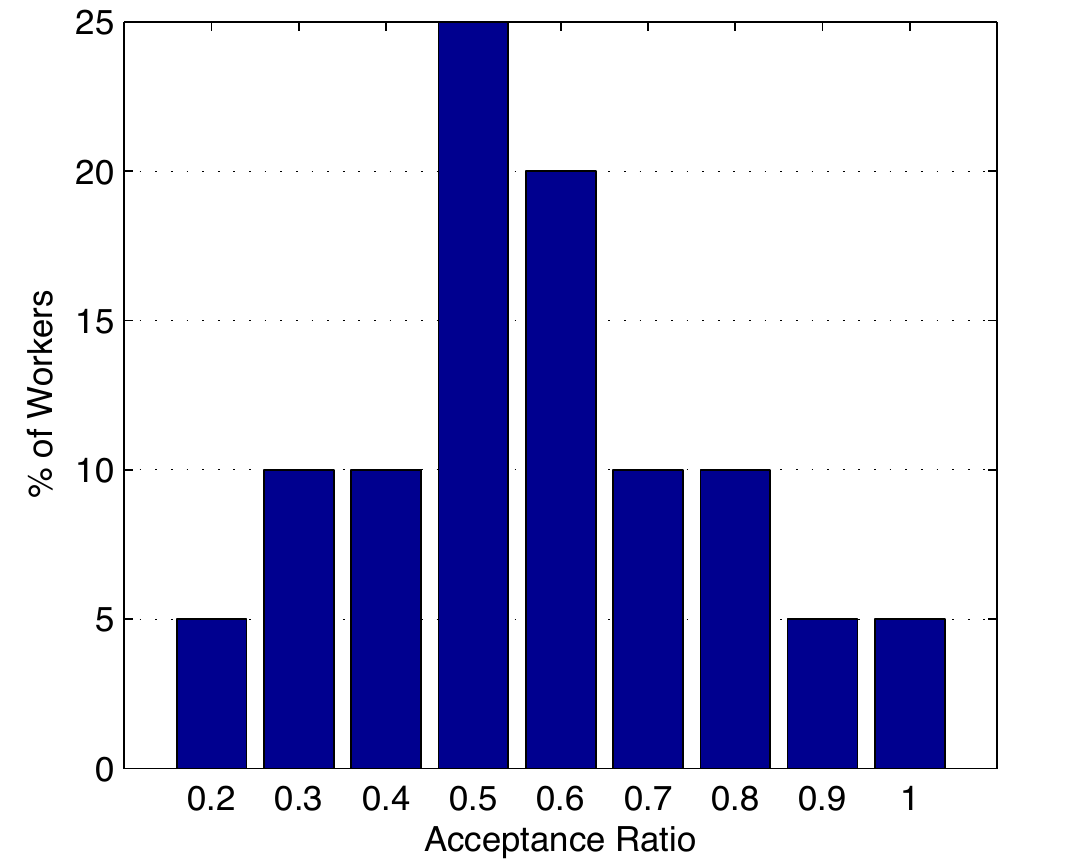}
\vspace{-0.05in}                 
                \caption{Acceptance ratio distribution}
                \label{fig:scn2}
        \end{subfigure}
        \begin{subfigure}[c]{0.23\textwidth}
                \includegraphics[width=\textwidth]{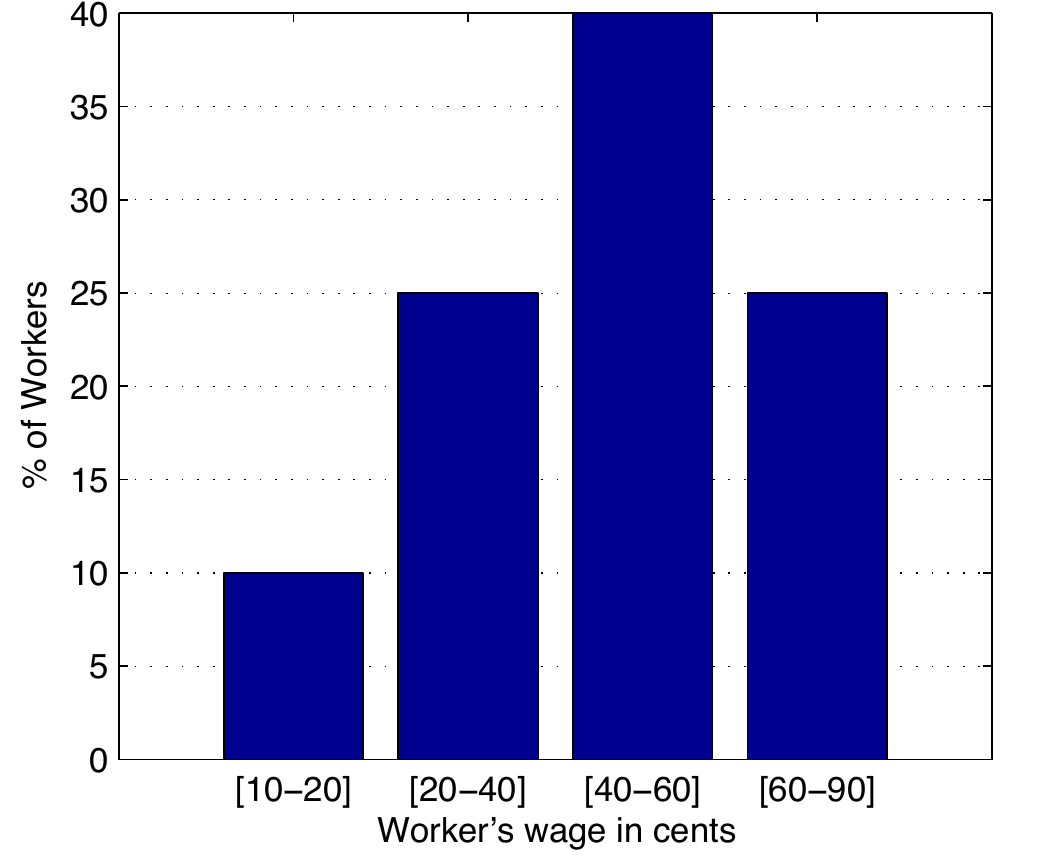}
\vspace{-0.05in}     
                \caption{Wage distribution}
                \label{fig:scn3}
        \end{subfigure}
         \begin{subfigure}[c]{0.23\textwidth}
                \includegraphics[width=\textwidth]{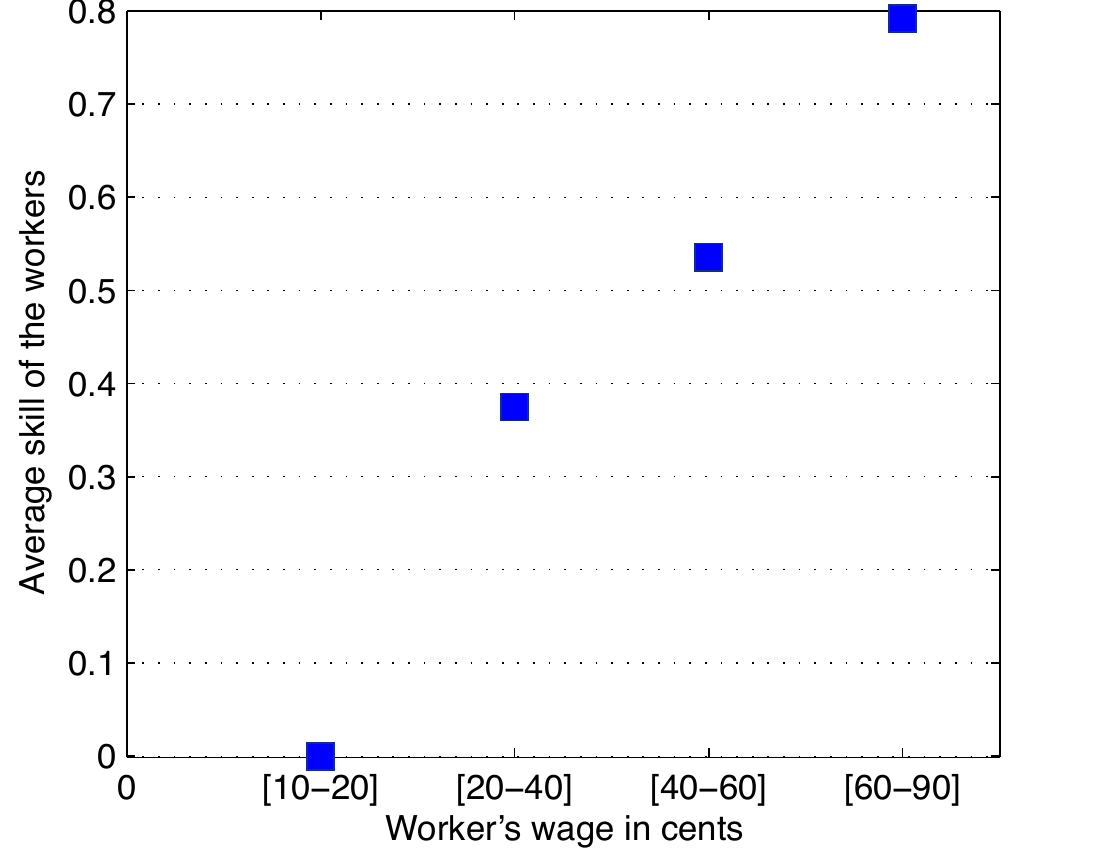}
\vspace{-0.05in}                 
                \caption{Strong positive correlation between worker skill and wage}
                \label{fig:scn4}
        \end{subfigure}
\vspace{-0.05in}        
       \caption{AMT worker profile distributions for the Egypt task}\label{fig:phase1}
\end{figure*}


\vspace{-0.05in}
\subsubsection{Stage 2 - Worker-to-Task Assignment}
In this stage, a subset ($56$ out of the $100$) of the workers among those who participated in Stage 1 is selected according to 3 worker-to-task assignment strategies: {\tt SmartCrowd}, {\tt Benchmark} and {\tt Online-Greedy}, as presented above. The minimum skill requirement per task is considered to be $1.8$, the maximum wage $\$2.0$ and $W_1=W_2=0.5$. The selected workers for each task are provided with a Google doc to collaboratively compose an article on the task's topic up to 150 words and in a time window of $24$ hours. The workers are suggested to use the answers of the Stage-1  questionnaires, as reference and/or starting point of their work 
Workers are also asked to care for quality aspects of their article, such as language correctness and information completeness. The final outcome of this stage is a production of $3$ documents per task, and a total of $15$ documents.
\vspace{-0.05in}
\subsubsection{Stage 3 - Task Evaluation}
KI-C evaluation is a delicate topic because it is objective. An appropriate technique for such objective evaluation is to again leverage the \emph{wisdom of the crowds}. This way a diverse and large enough group of individuals can accurately evaluate information to nullify individual biases and herding effect. Therefore, we crowdsource \emph{the task evaluation}. Each completed task (set of $3$ documents) from Stage 2 is set up as a HIT in AMT, and $30$ workers are assigned to evaluate it considering 5 key quality assessment aspects~\cite{quality}, without knowing the underlying task production algorithm. The results listed in Table~\ref{tab:userstudy} indicate that the use of {\tt SmartCrowd} indeed leads to more qualitative KI-C tasks, across \emph{all} of the measured quality axes. 

\begin{table*}[t]
\begin{tabular}{ |l|l|l|l|l|l|l|l| }
\hline
\multicolumn{8}{ |c| }{Average Rating} \\
\hline
Task & Algorithm & Completeness & Grammar & Neutrality & Clarity & Timeliness & Added-value \\ \hline
\multirow{3}{*}{Egypt political unrest} & {\tt SmartCrowd} & 4.5 & 4.2 & 4.0 & 4.2 & 4.1 & 4.0 \\
& {\tt Online-greedy} & 3.3 & 3.4 & 3.3 & 3.0 & 3.4 & 3.1\\
& {\tt Benchmark} & 3.1 & 3.2 & 3.3 & 3.1 & 3.2 & 2.9 \\
\hline
\multirow{3}{*}{NSA document leakage} & {\tt SmartCrowd} & 4.5 & 4.7 & 4.3 & 3.9 & 4.1 & 4.1\\
& {\tt Online-greedy} & 3.2 & 3.4 & 3.3 & 3.3 & 3.0 & 2.9\\
& {\tt Benchmark} & 3.3 & 3.3 & 3.4 & 2.9 & 2.9 & 3.4 \\
\hline
\multirow{3}{*}{PS Games} & {\tt SmartCrowd} & 4.3 & 4.2 & 4.0 & 4.1 & 4.1 & 4.2 \\
& {\tt Online-greedy} & 3.2 & 3.3 & 3.3 & 3.1 & 3.0 & 2.9\\
& {\tt Benchmark} & 3.0 & 3.2 & 3.1 & 2.8 & 2.9 & 2.9 \\
\hline
\multirow{3}{*}{All electric cars} & {\tt SmartCrowd} & 4.2 & 4.2 & 4.1 & 4.4 & 4.0 & 4.1 \\
& {\tt Online-greedy} & 3.0  & 3.1 & 3.3 & 3.0 & 2.9 & 2.8 \\
& {\tt Benchmark} & 2.9 & 2.6 & 2.6 & 3.0 & 2.8 & 2.3  \\
\hline
\multirow{3}{*}{Global warming} & {\tt SmartCrowd} & 4.2 & 4.3 & 4.5 & 4.2 & 4.1 & 3.7 \\
& {\tt Online-greedy} & 3.0 & 3.2 & 3.1 & 3.4 & 3.3 & 3.3\\
& {\tt Benchmark} & 2.9 & 2.9 & 3.1 & 3.2 & 2.9 & 2.7 \\
\hline
\end{tabular}
\caption{\small Quality assessment is performed (scale $1-5$) by a new set of $150$ AMT workers of the composed write-ups in Stage-2. {\tt SmartCrowd} always outperforms the other two alternatives consistently and significantly across all quality dimensions. Interestingly,  {\tt Benchmark} is somewhat comparable to {\tt Online-greedy} in most of the quality dimensions (only slightly less). As workers are self-appointed to the tasks in {\tt Benchmark}, in our post analysis we observe that some expert workers are assigned to the task, increasing the overall quality. On an average $\$1.81$ is spent per task for {\tt Online-greedy}, whereas, $\$1.936$ and $\$1.84$ are spent on {\tt SmartCrowd} and {\tt Benchmark} respectively. These results corraborate that the proposed optimization in \sys\ for knowledge-intensive crowdsourcing is {\em effective} to achieve high quality results in a cost-effective way.}\label{tab:userstudy}
\end{table*}

\subsection{Synthetic Data Experiments}
These experiments are conducted on an Intel core i7 CPU, 8 GB RAM machine. IBM CPLEX version 12.5.1 is used for solving the ILP. An event-based simulator is designed on Java Netbeans to simulate the crowdsourcing environment. All results are presented as the average of 3 runs.


{\bf Simulator Parametrization:} The distribution of the parameters presented below are chosen akin to their respective distributions, observed in our real AMT populations. \\
\noindent 1. {\em Simulation Period} - We simulate the system for a time period of 10 days, i.e. 14400 simulation units, with each simulation unit corresponding to 1 minute.\\
\noindent 2. {\em \# of skills} - a total of $|\mathcal{S}|$= 10 skills are simulated. Unless otherwise stated, the default  \# of skills in a task is  $1$. \\
\noindent 3. {\em \# of Workers} - $|\mathcal{U}|$= 10,000.\\
\noindent 4. {\em Profile of a worker} - $u_{s_i}$ in skill $s_i$ receives a random value from a normal distribution with the mean set to $0.5$, variance $0.15$. $w_u$ receives a random value from a normal distribution with a mean set to $0.5$, variance $0.2$. $p_u$ is also normal with a a mean set to $0.5$, variance $0.1$.\\
\noindent 5. {\em Tasks} -  A normal variable with mean $15$, variance $3$ is multiplied with another normal random variable with mean $0.7$, variance $0.15$ to get $Q_{t_i}$, whereas, the former normal random variable is multiplied  with a different normal random variable with mean $0.5$, variance $0.2$ to get $W_t$.\\
\noindent 6. {\em Weights} - Unless otherwise stated, $W_1=W_2=0.5$.\\
\noindent 7. {\em Worker Arrival, Task Arrival} - Workers arrive  following a Poisson process, with an arrival rate of $\mu =10$/minute.
Tasks arrive also in a Poisson distribution with an arrival rate of $\kappa=20$/minute.\\
\noindent 8. {\em Workload} - Unless otherwise stated, the workload is designed with $10,000$ tasks.

{\bf Implemented Algorithms:}
{\tt Benchmark:} It models a typical crowdsourcing environment, where the workers are self-appointed to the tasks, trying to maximize their individual profit. The algorithm also performs worker pre-filtering, similar to the pre-qualification tests used by today's crowdourcing platforms, allowing workers to undertake a certain task $t$ only if their skill is above 10\% of the task's skill requirement $Q_{t_i}$.  \\
{\tt Online-Greedy:} As soon as a worker arrives, it finds from the available tasks the ones that pay more than the worker's minimum wage. Then it calculates the worker's marginal utility on the filtered tasks and suggests worker the task with the highest utility. This algorithm is an adaptation of one of the latest state-of-the-art strategies for online task assignment~\cite{chienJuHo}.\\
{\tt Online-Optimal:} It optimally solves the ILP problem of Equation~\ref{eqn:eq} in a purely online fashion; when invoked, it uses only the workers that are currently online on the tasks that currently require worker assignment. \\ 
{\tt C-DEX:} generates an optimal solution (Section~\ref{alg1}).\\
{\tt Offline-CDEX-Approx, Online-CDEX-Approx:} generates an approximate solution for offline computation and online maintenance (Section~\ref{greedycdex}).\\
{\tt C-DEX$^+$:} generates an approximate solution (Section~\ref{alg2}).
\vspace{-0.1in}
\subsubsection{Performance Experiments}
We design experiments for: Offline phase (index building) and  Online phase (index maintenance). Two measures are used: clock time for the index building and maintenance stages, and the fraction of successful tasks for the worker-to-task assignment stage ($\frac{\text{\# of succesful task assignments}}{\# tasks}$). 
\vspace{-0.1in}
\paragraph{Index Building (offline)}\label{precomp}
We vary the workload size of {\tt \Index},{\tt \vindex} and \\ {\tt Offline-CDEX-Approx} with $|\mathcal{U}|=10,000$, and measure clock time for index computation (in minutes). Recall that {\tt \vindex} needs to have the Virtual Worker set ($\mathcal{N}$) computed first. For that, our experimental evaluation sets $\alpha$ to $20$-th percentile pair-wise Euclidean distance in ascending order, and observes that the computation time is within $2$ minutes, resulting in $|\mathcal{N}|=620$ Virtual Workers. The results are presented in {\bf Figure~\ref{fig:1} (consider the primary Y-axis)}. Unsurprisingly, {\tt Offline-CDEX-Approx} is the fastest among the three alternatives, but \vindex\ is very comparable. Beyond $50,000$ tasks, {\tt Index} fails to respond. 
\vspace{-0.1in}
\paragraph{Worker Replacement (online)}\label{taskp}
We compare the six implemented algorithms. The index-based algorithms ({\tt \Index},{\tt \vindex} and {\tt Online-CDEX-Approx}) become clear winner compared to the rest.


{\bf Simulation period - Figures \ref{fig:2}, \ref{fig:2a}, \ref{fig:3}:} In Figures \ref{fig:2} and \ref{fig:2a}, we measure system performance (fraction of successful tasks) throughout the simulation period at discrete intervals (every 2 days). Figure~\ref{fig:2a} captures the special case with $W_2=0, X_l=0$ and zero skill threshold. Note that, under this condition {\tt Online-CDEX-Approx} has a provable approximation factor.  We can observe that the proposed index-based strategies outperform the remaining ones significantly and that they maintain their throughput over the entire simulation period, while the other algorithms peak and then drop midway, as a result of their myopic worker-task assignment decisions that penalize the overall outcome. However, Figure \ref{fig:2} and \ref{fig:2a} still depict a better-than-reality performance for some algorithms, since certain bad assignments are not counted as such due to the measurement discretization. For example, if a task comes at time unit 1, languishes until time 2388 before getting assigned, it will still count as a successful task. Figure \ref{fig:3} investigates this behavior by measuring average task end-to-end time, i.e. the difference in  time between a task arrival and the time when a set of workers satisfying the task's quality/cost requirements have accepted to take it. This measurement is taken only for successful tasks and smaller is better. It can be observed that our proposed algorithms finish in less than 2 time units mainly because of our worker replacement strategy. The other algorithms including {\tt Online-Greedy} take significantly more time that justifies the necessity of pre-computation.

{\bf Vary the ratio of task to worker arrival rate - Figure~\ref{fig:4}:}
All algorithms perform well when the ratio of task arrival to worker arrival rate is small, because of the oversupply of workers. However, with high task arrival rate index based strategies outperform all the remaining solutions.


{\bf Vary \# skills/task - Figure~\ref{fig:5}:}
As skills per task increase, the fraction of successful tasks decreases for all algorithms, since finding the right worker becomes harder in a high-dimensional task/worker setting. Nevertheless, the index-based strategies still manage to keep a steadily high performance, outperforming all remaining ones.


{\bf Vary acceptance ratio - Figure~\ref{fig:6}:}
With high acceptance ratio, performance improves in general, as workers become more predictable. The index-based strategies consistently outperform the others. 


{\bf Vary mean skill - Figure~\ref{fig:7}:}
As expertise becomes scanty (i.e. low values of mean worker skill) {\tt Benchmark} and \\ {\tt Online-Greedy} perform very poorly as  they need to scan and seek more workers to reach the task skill threshold. This justifies that the optimization objective in \sys\ is meaningful for knowledge-intensive tasks.

\vspace{-0.1in}
\paragraph{Worker Addition, Deletion, Update (online)}\label{main}
We vary the \# of new workers, \# of deleted workers,and \# of workers with profile updates and measure the incremental maintenance time for {\tt \Index}, {\tt \vindex}, and {\tt Online-CDEX-Approx}. The results  for worker addition are presented in {\bf Figure~\ref{fig:15}}. The deletion and update cases give similar results and are omitted for brevity. Results show that our incremental index maintenance techniques are efficient. However, the approximate solutions warrant higher efficiency compared to the optimal one.

\vspace{-0.1in}
\subsubsection{Quality Experiments}
For the quality simulation experiments we measure the value of the normalized objective function.
\vspace{-0.1in}
\paragraph{Index Building (offline)}\label{precompQ}
The setting is akin to Section~\ref{precomp}, but here we measure the objective function value instead. The results ({\bf consider the secondary Y-axis of Figure~\ref{fig:1}}) demonstrate that both approximation algorithms {\tt \vindex}and {\tt Offline-CDEX-Approx} return high quality solutions that are comparable to its optimal counterpart {\tt \Index}.
\vspace{-0.1in}
\paragraph{Worker Replacement (online)}
{\bf Simulation period - Figures~\ref{fig:9} and \ref{fig:9a}} have similar settings that of Figure~\ref{fig:2} and \ref{fig:2a}. Our proposed index-based strategies significantly outperform the others throughout the period of the simulation. As expected, {\tt Benchmark} performs the worst. {\tt Online-CDEX-Approx} returns higher quality in Figure~\ref{fig:9a} as the algorithm guarantees a provable approximation factor under that settings.




{\bf Vary \# skills-Figure~\ref{fig:10}:}
The index-based strategies  outperform all remaining ones, even for tasks that require multiple skills, similarly to Figure~\ref{fig:5}. 


{\bf Vary acceptance ratio - Figure~\ref{fig:11}:}
The index-based strategies  {\tt \Index}, {\tt \vindex}, and {\tt Online-CDEX-Approx} outperform all the remaining ones, even with small mean worker acceptance ratio.


{\bf Vary mean skill - Figure~\ref{fig:12}:}
The index-based strategies  consistently win over the rest, including the case where expertise is very scarce.

{



{\bf Vary $W_1, W_2$ - Figure~\ref{fig:14}:} 
As expected, when $W_1$ increases,  all algorithms seek to improve quality more than cost and task quality increases. The index based solutions outperform the rest of the competitors with high $W_1$ (task that require optimization over skills), compared to the rest.  

 \vspace{-0.1in}
\paragraph{Worker Addition, Deletion, Update (online)}\label{onlineq}
It considers a similar settings as Experiment~\ref{main}. We observe that our index based approximate solutions ({\tt Online-CDEX-Approx} and \vindex) are comparable to the optimal solution \Index\ in quality. The results are omitted for brevity.

\begin{figure*}
\centering
\begin{minipage}[t]{0.30\textwidth}
    \includegraphics[height=30mm, width = 60mm]{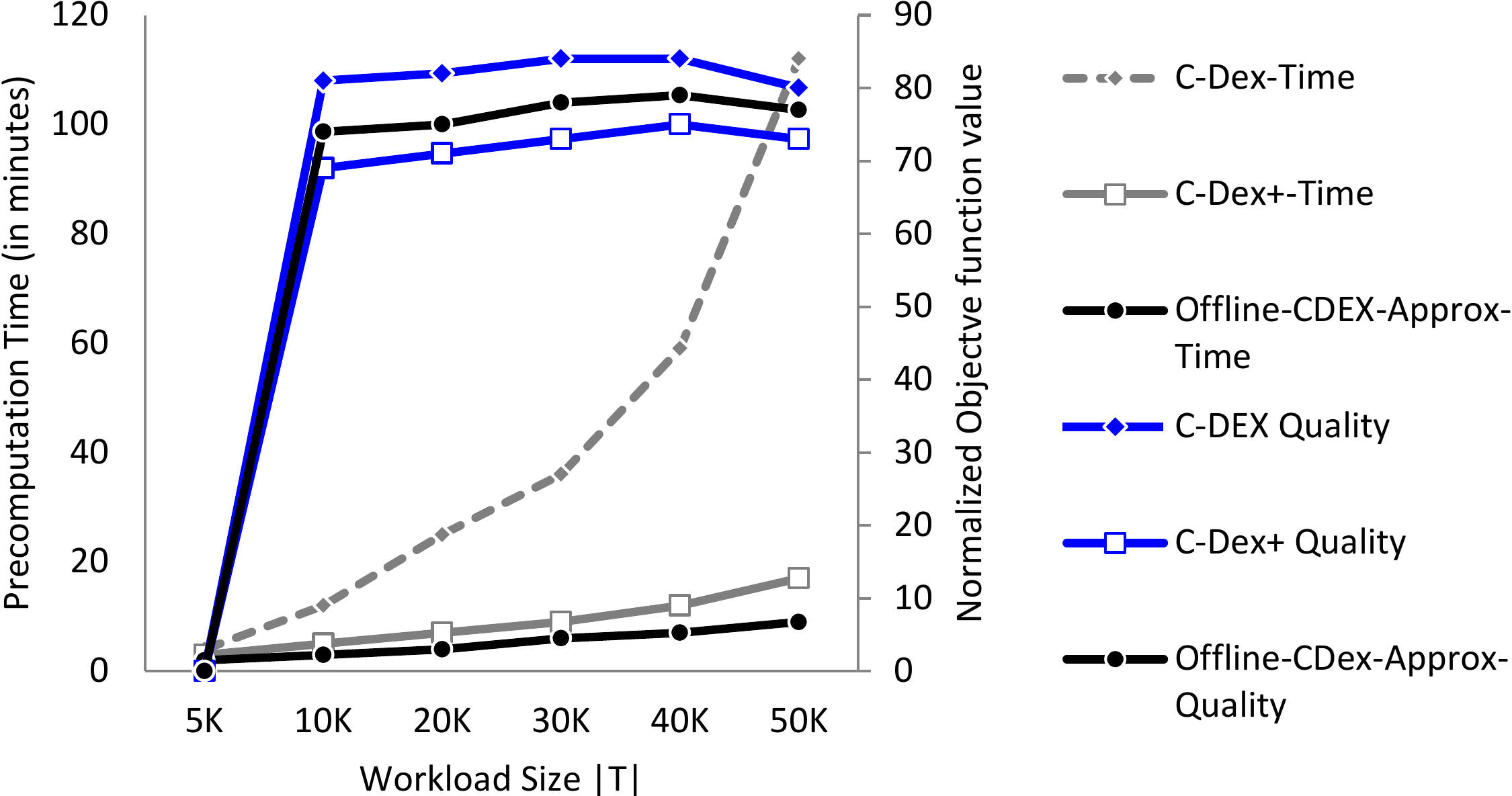}
\vspace{-0.05in}     
    \caption{\small{Index Building Time and Quality varying workload}}
    \label{fig:1}
\end{minipage}
\hspace{5mm}
\begin{minipage}[t]{0.30\textwidth}
    \includegraphics[height=30mm, width = 60mm]{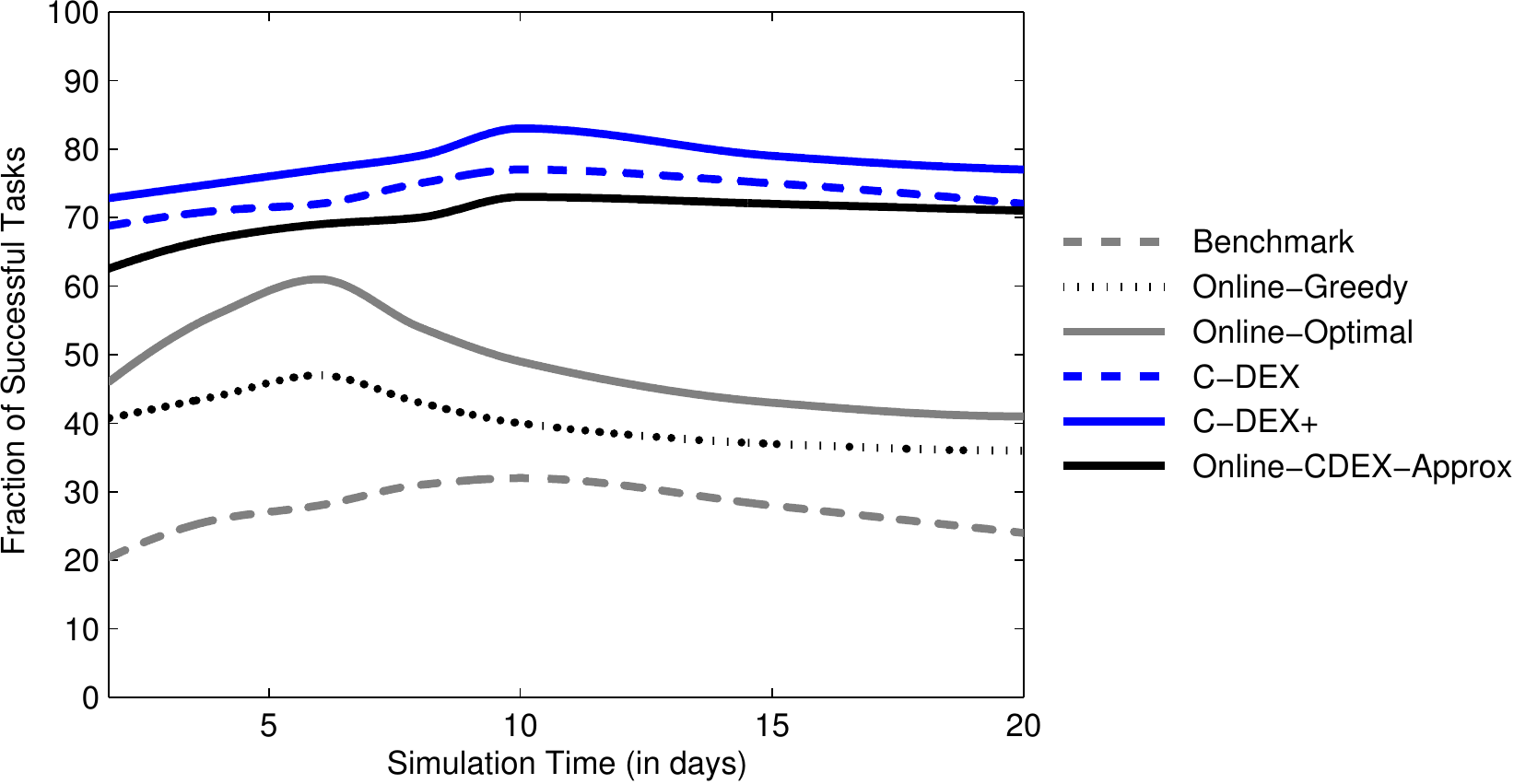}
\vspace{-0.05in}  
   \caption{\small{Performance varying simulation time}}
    \label{fig:2}
\end{minipage}
\hspace{5mm}
\begin{minipage}[t]{0.30\textwidth}
\centering
   \includegraphics[height=30mm, width = 60mm]{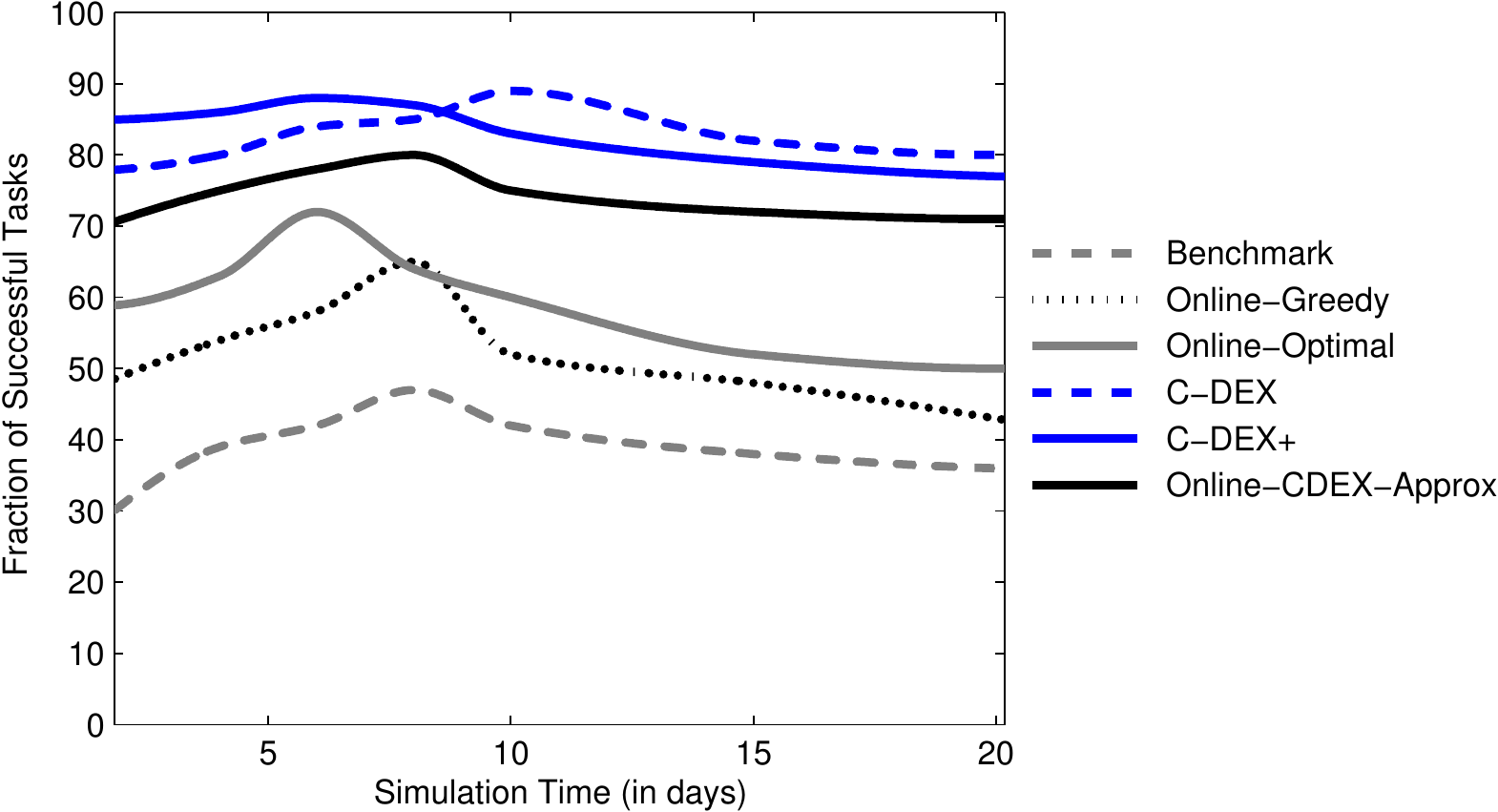}
  \vspace{-0.05in} 
    \caption{\small{Performance varying simulation time with no skill threshold and $W_2=0,X_l=0$}}
    \label{fig:2a}
\end{minipage}
\end{figure*}

\begin{figure*}
\begin{minipage}[t]{0.30\textwidth}
    \includegraphics[height=30mm, width = 60mm]{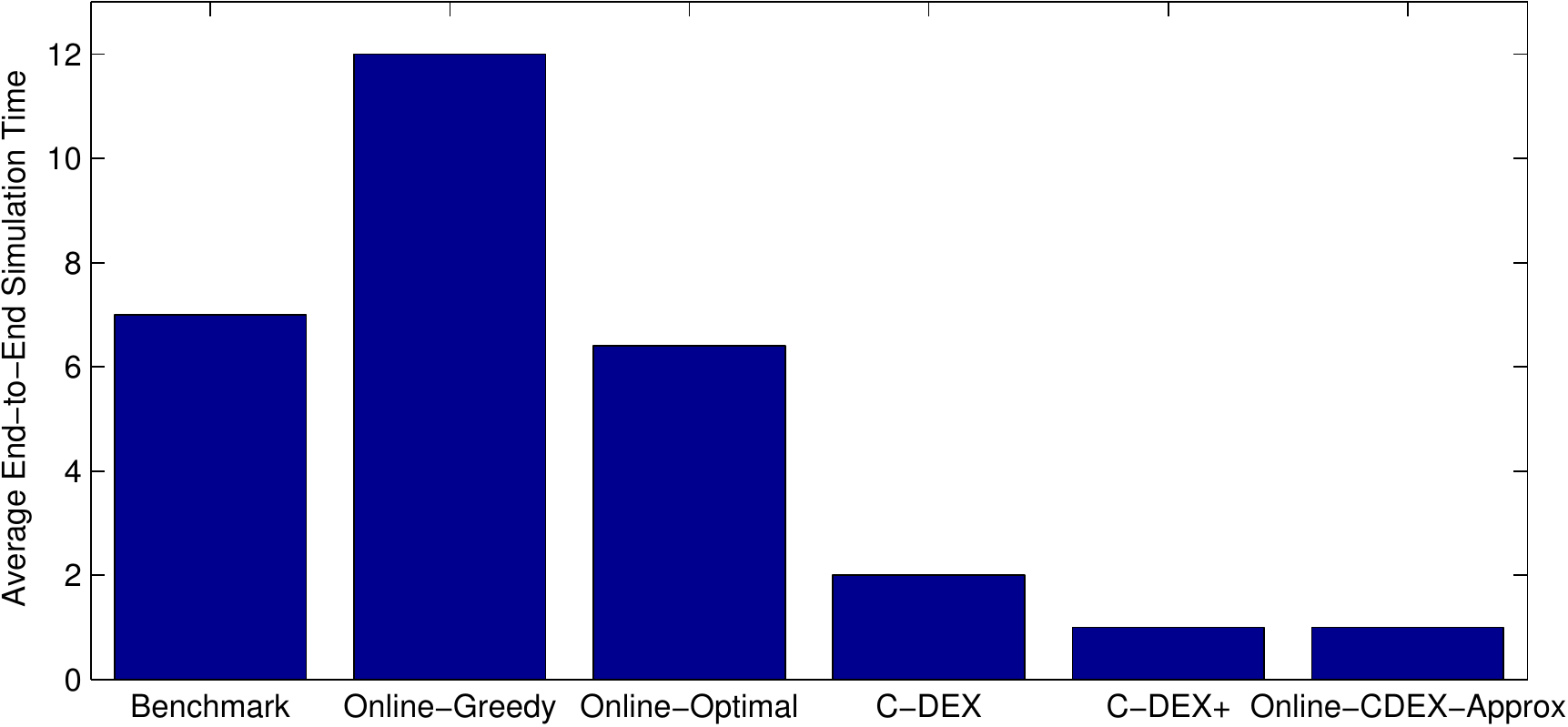}
\vspace{-0.05in}    
    \caption{\small{Performance after entire simulation period}}
    \label{fig:3}
\end{minipage}
\hspace{5mm}
\begin{minipage}[t]{0.30\textwidth}
\includegraphics[height=30mm, width = 60mm]{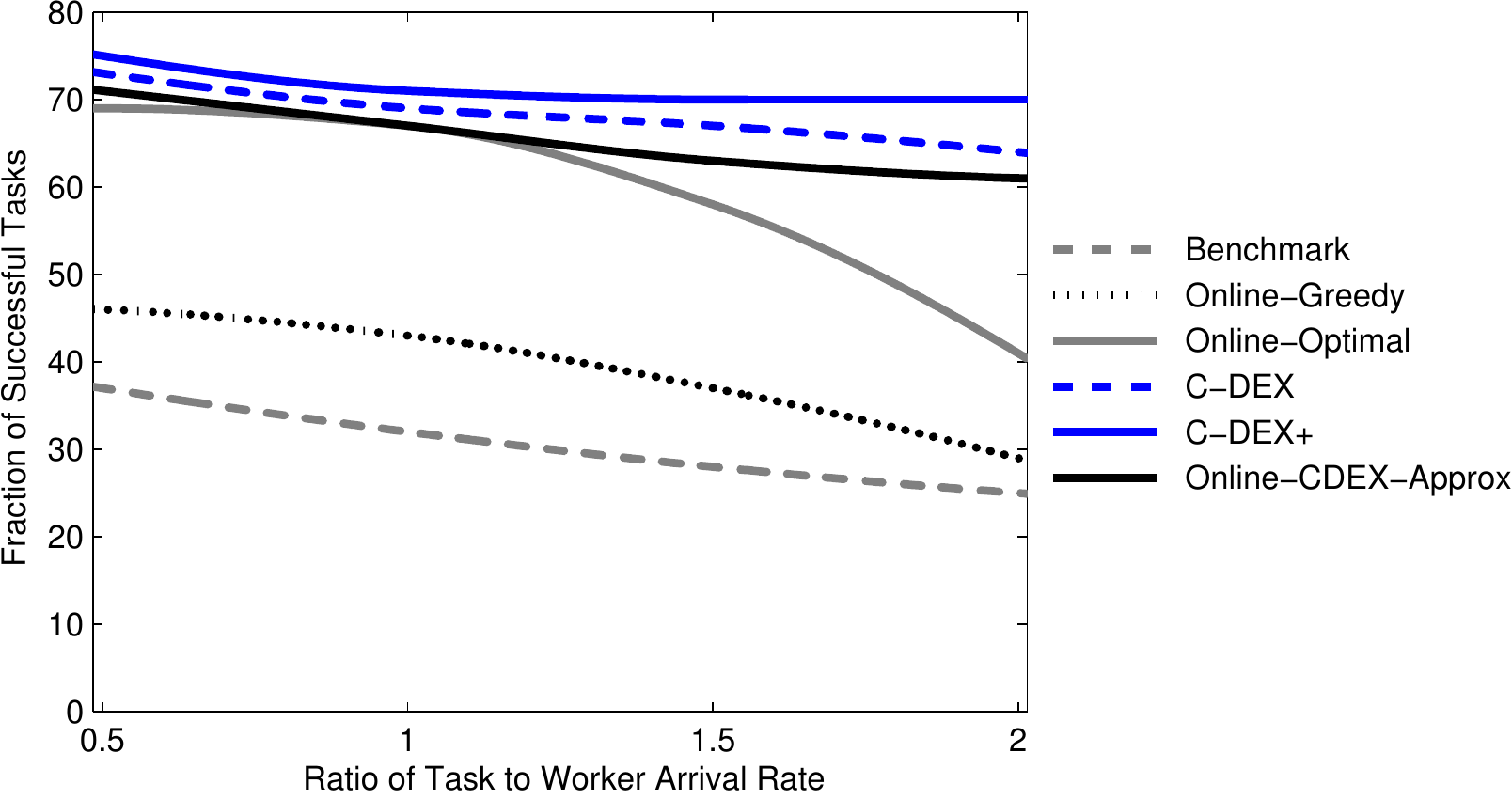}
\vspace{-0.05in}    
    \caption{\small{Performance varying the ratio of task to worker arrival rate}}
    \label{fig:4}
\end{minipage}
\hspace{5mm}
\begin{minipage}[t]{0.30\textwidth}
    \includegraphics[height=30mm, width = 60mm]{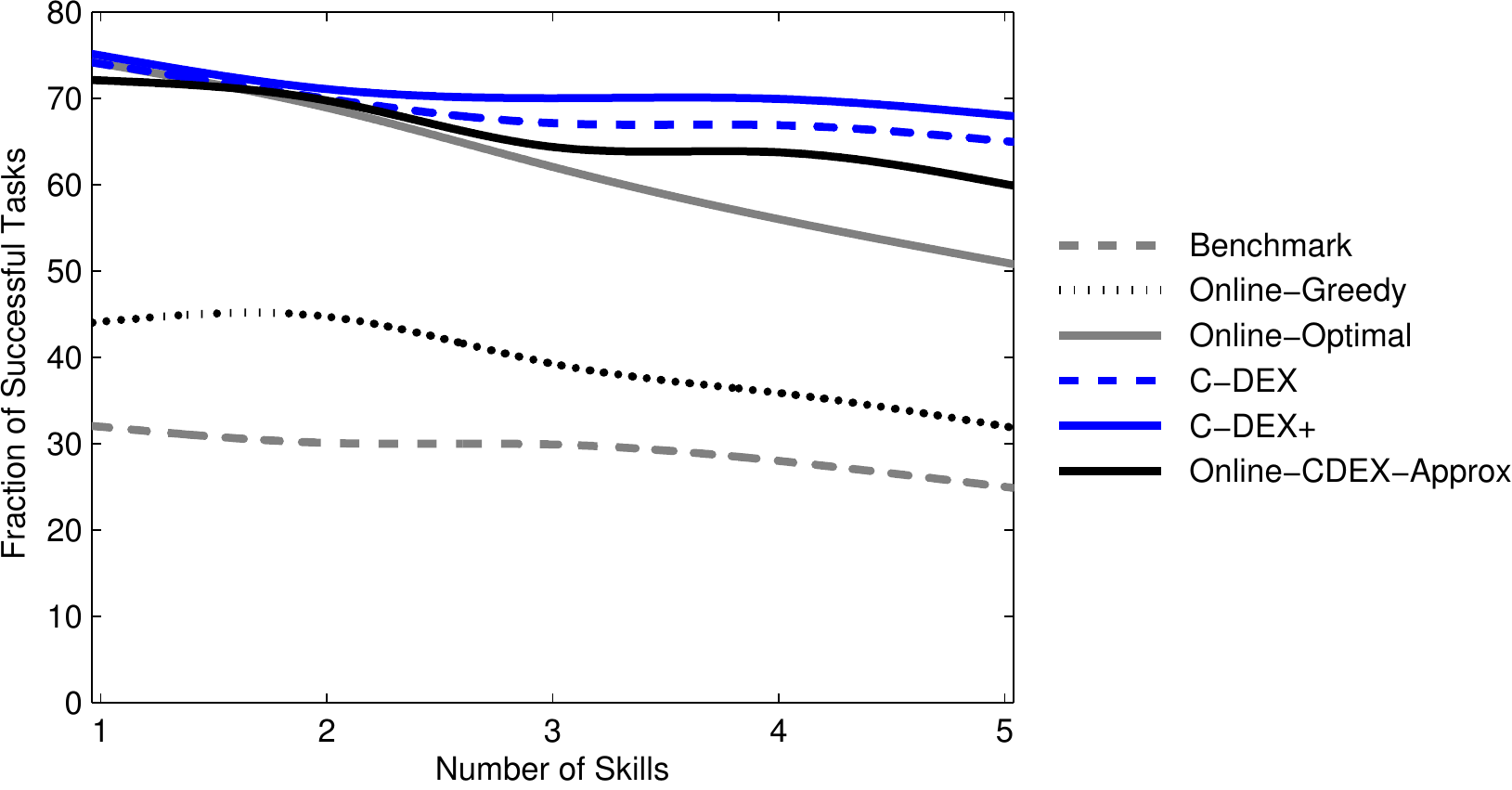}
\vspace{-0.05in}   
   \caption{\small{Performance varying \# of skills/task}}
    \label{fig:5}
\end{minipage}
\end{figure*}

\begin{figure*}
\centering
\begin{minipage}[t]{0.30\textwidth}
    \includegraphics[height=30mm, width = 60mm]{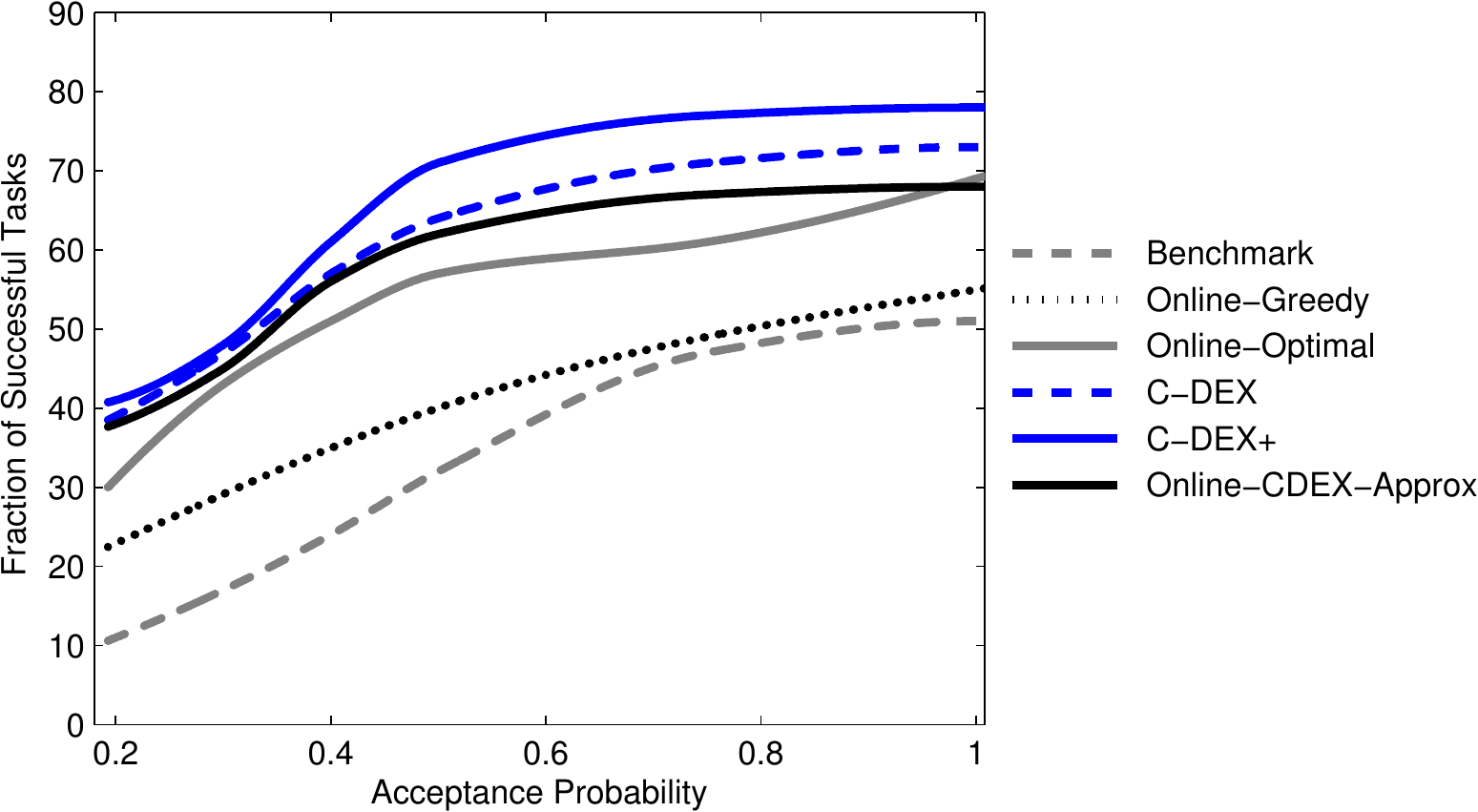}
\vspace{-0.05in}   
    \caption{\small{Performance varying acceptance ratio}}
    \label{fig:6}
\end{minipage}
\hspace{5mm}
\begin{minipage}[t]{0.30\textwidth}
 \includegraphics[height=30mm, width = 60mm]{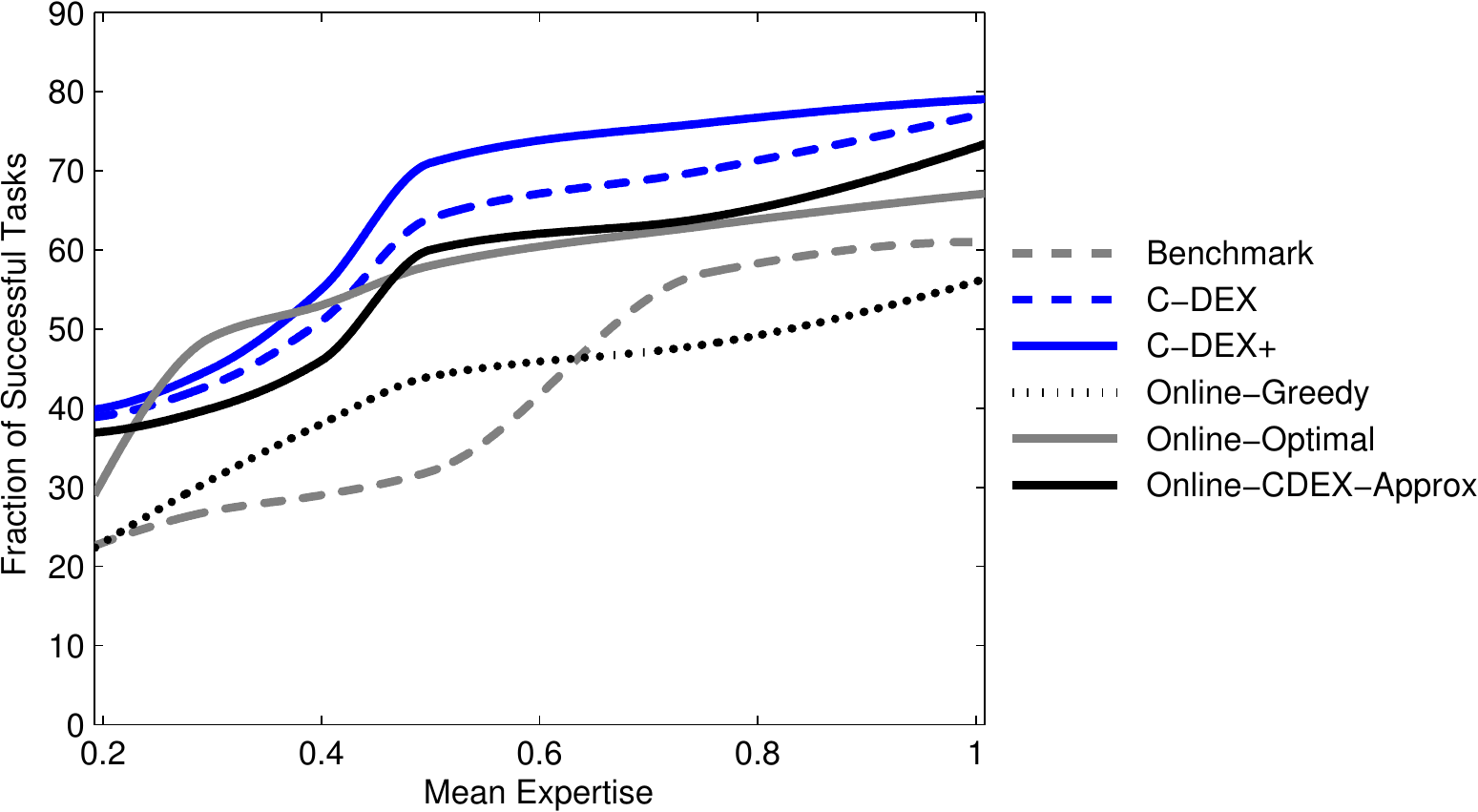}
\vspace{-0.05in}    
   \caption{\small{Performance varying mean skill}}
    \label{fig:7}
\end{minipage}
\hspace{5mm}
\begin{minipage}[t]{0.30\textwidth}
    \includegraphics[height=30mm, width = 60mm]{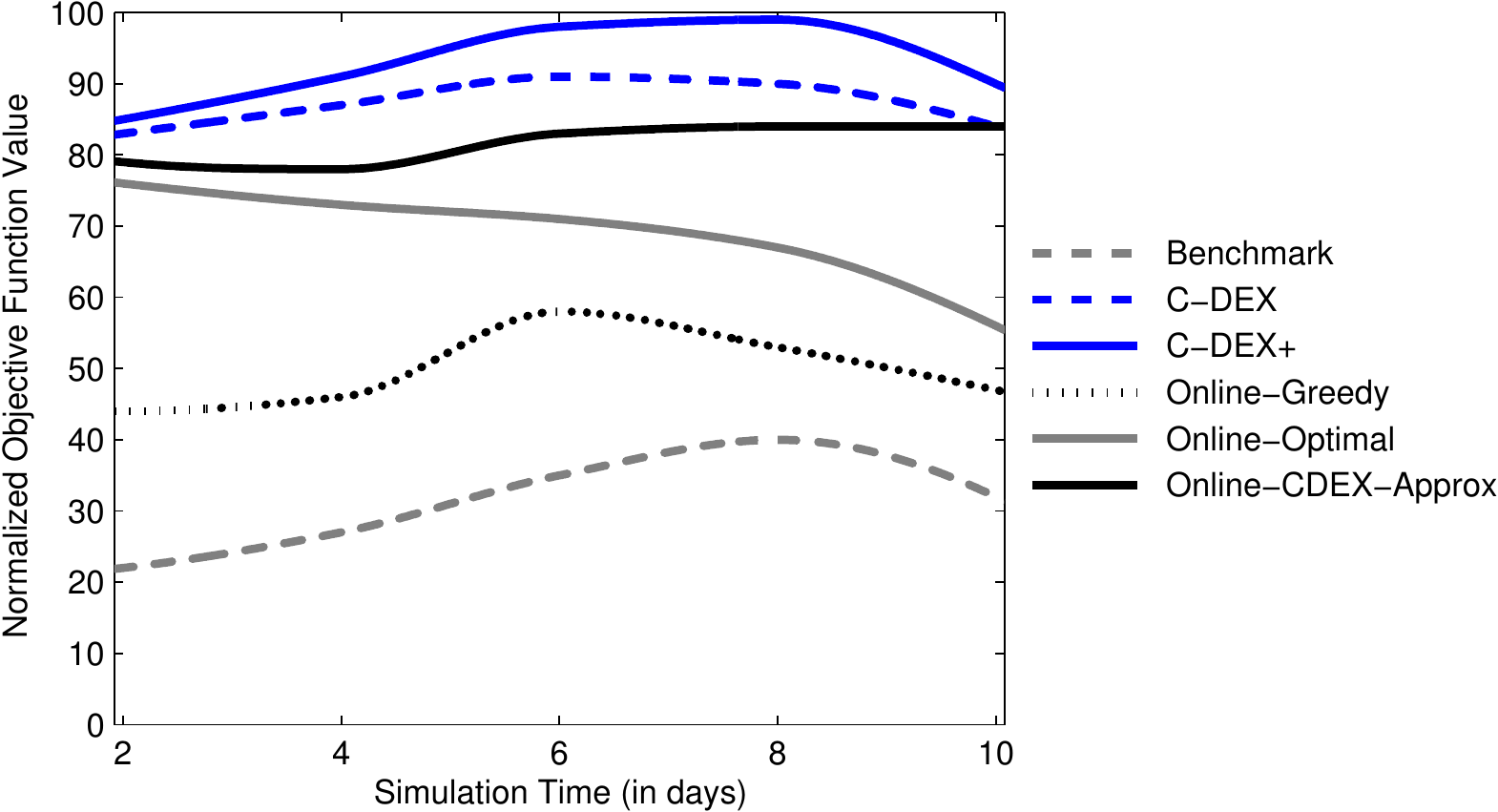}
\vspace{-0.05in}    
    \caption{\small{Objective function varying simulation time}}
    \label{fig:9}
\end{minipage}
\end{figure*}

\begin{figure*}
\centering
\begin{minipage}[t]{0.30\textwidth}
    \includegraphics[height=30mm, width = 60mm]{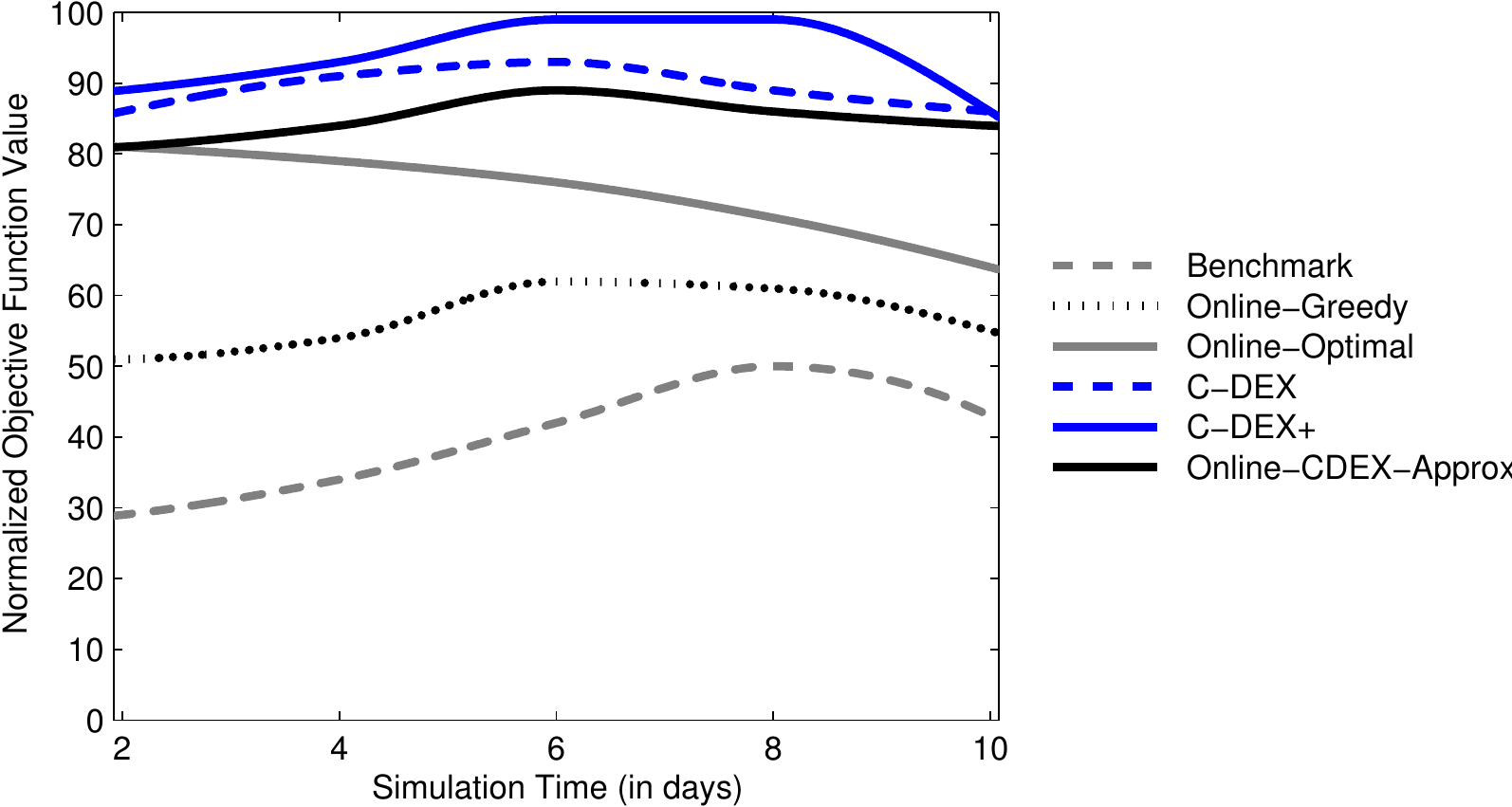}
   \caption{\small{Objective function varying simulation time with no skill threshold $W_2=0,X_l=0$}}
   \label{fig:9a}
\end{minipage}
\hspace{5mm}
\begin{minipage}[t]{0.30\textwidth}
    \includegraphics[height=30mm, width = 60mm]{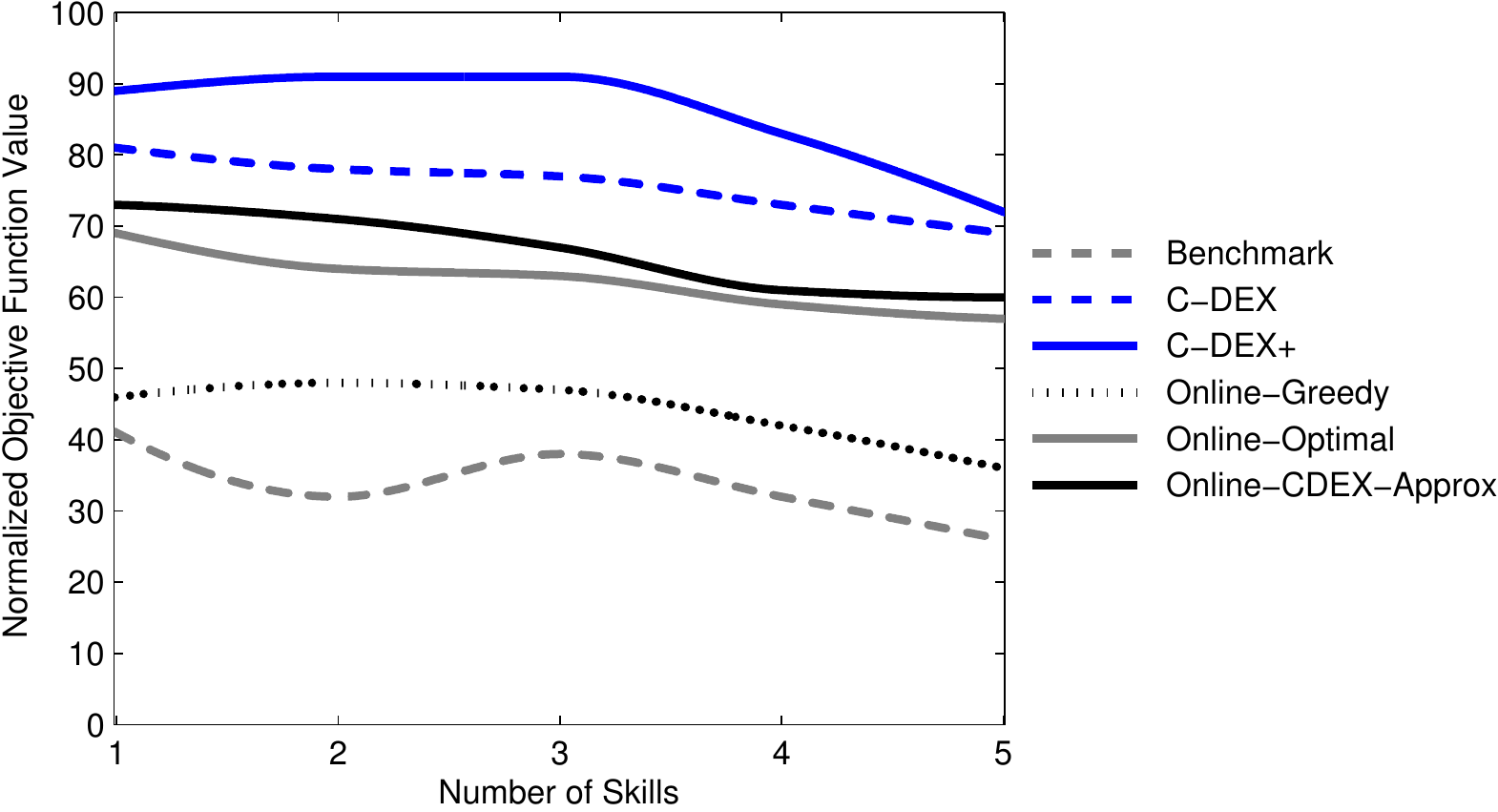}
    \caption{\small{Objective function varying \# of skills/task}}
    \label{fig:10}
\end{minipage}
\hspace{5mm}
\begin{minipage}[t]{0.30\textwidth}
    \includegraphics[height=30mm, width = 60mm]
     {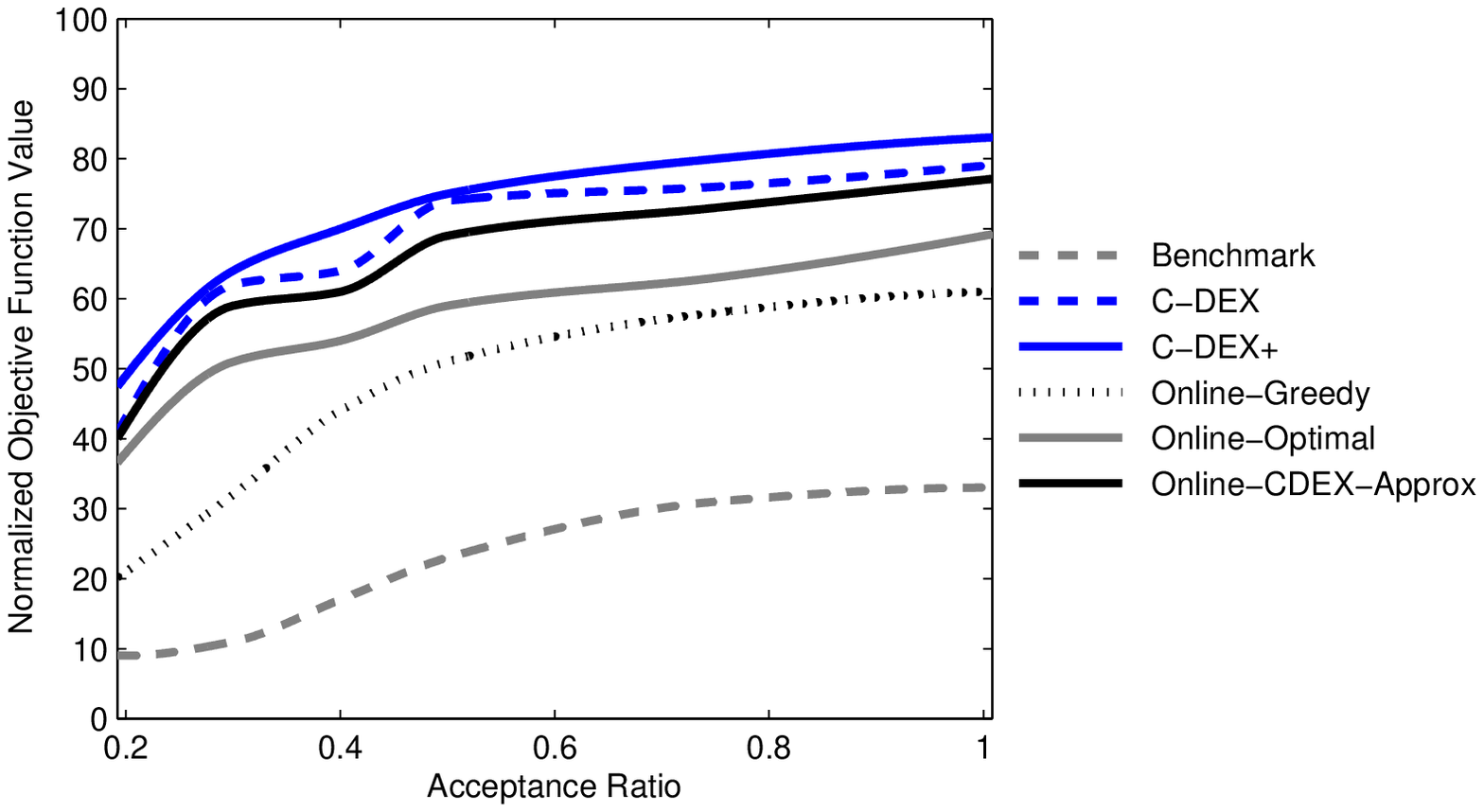}
    \caption{\small{Objective function varying acceptance ratio}}
    \label{fig:11}
\end{minipage}
\end{figure*}

\begin{figure*}
\centering
\begin{minipage}[t]{0.30\textwidth}
 \includegraphics[height=30mm, width = 60mm]
    {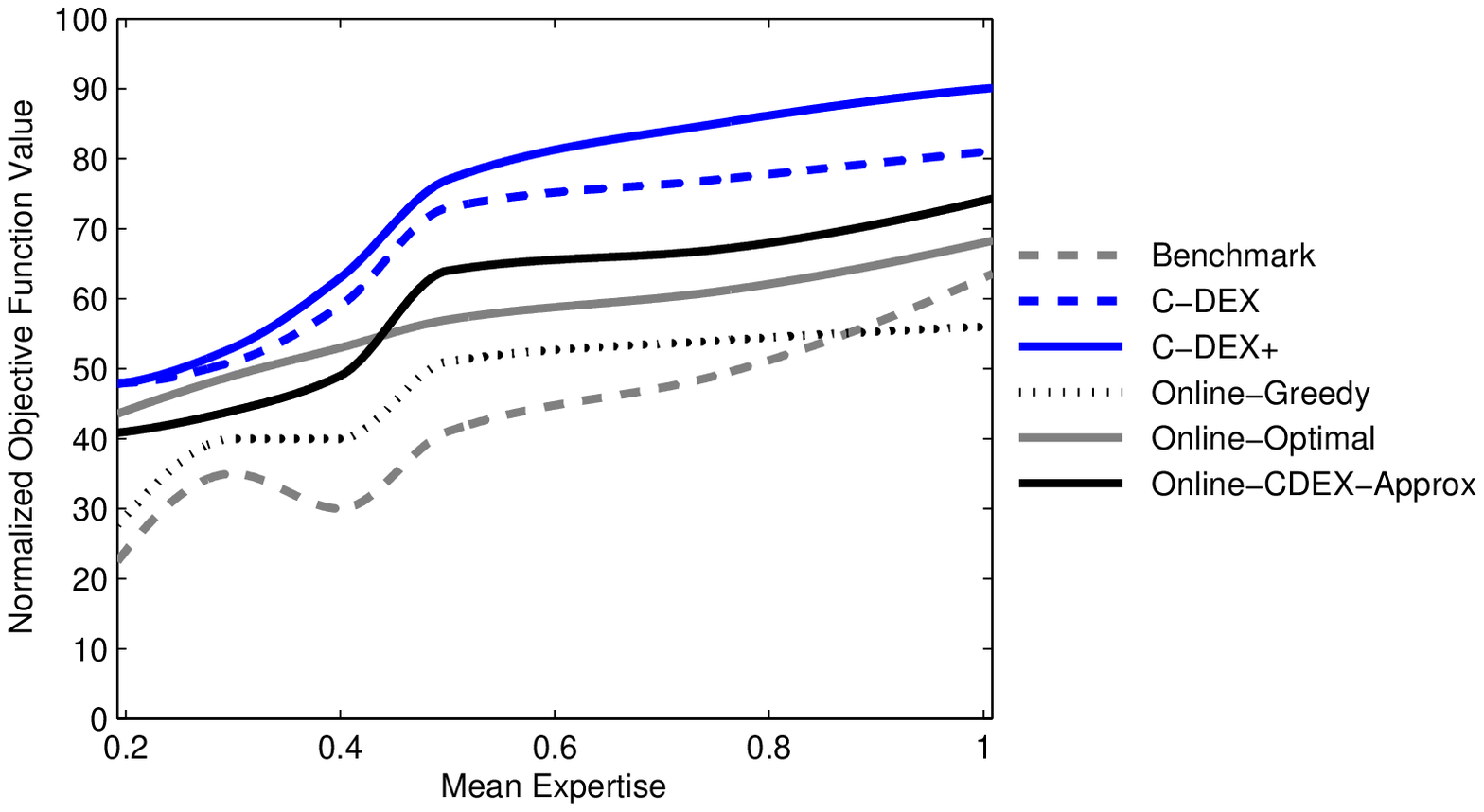}
    \caption{\small{Objective function varying mean skill}}
    \label{fig:12}
\end{minipage}
\hspace{5mm}
\begin{minipage}[t]{0.30\textwidth}
    \includegraphics[height=30mm, width = 60mm]{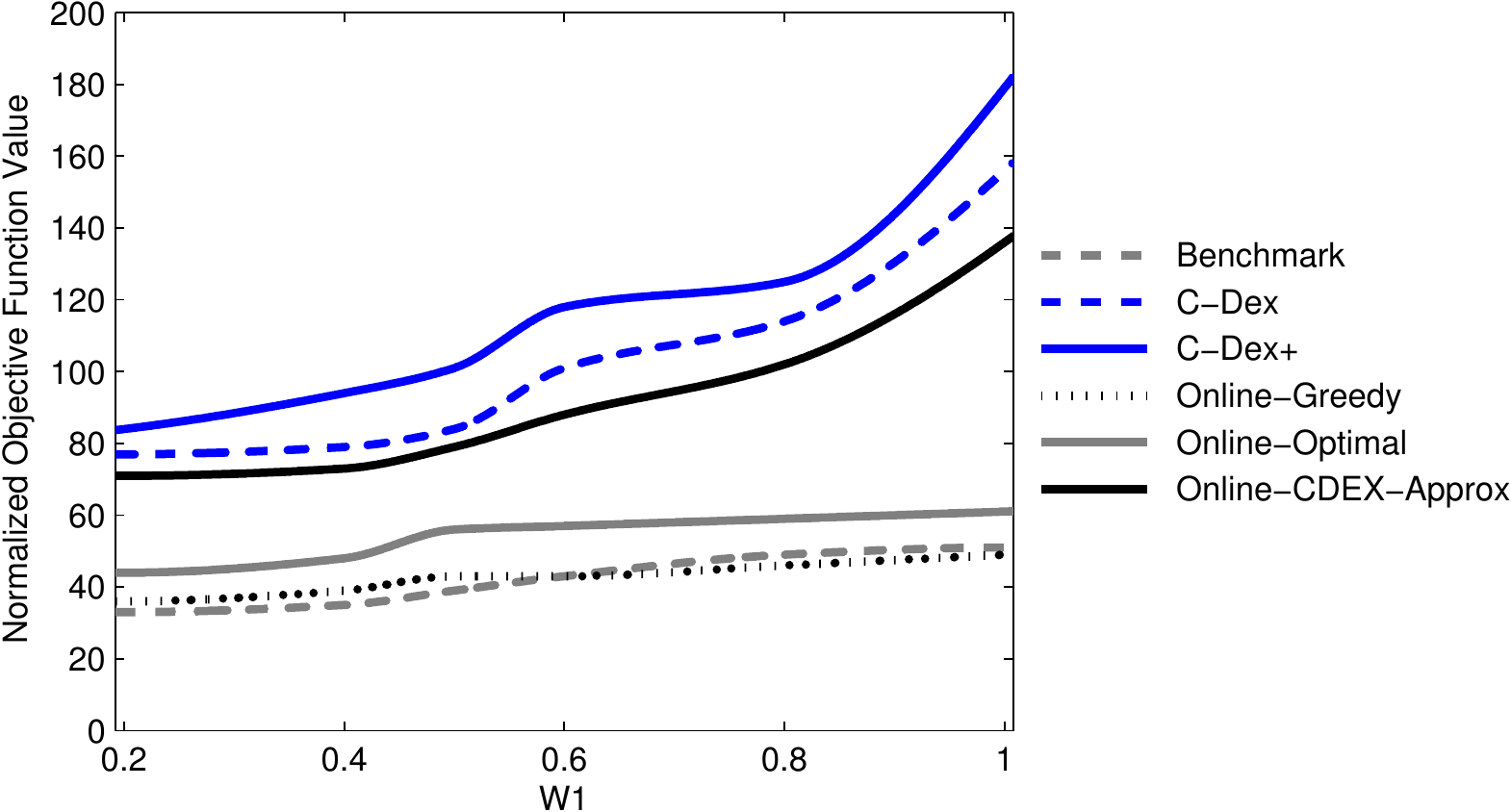}
    \caption{\small{Objective function varying $W_1,W_2$}}
    \label{fig:14}
\end{minipage}
\hspace{5mm}
\begin{minipage}[t]{0.30\textwidth}
    \includegraphics[height=30mm, width = 60mm]{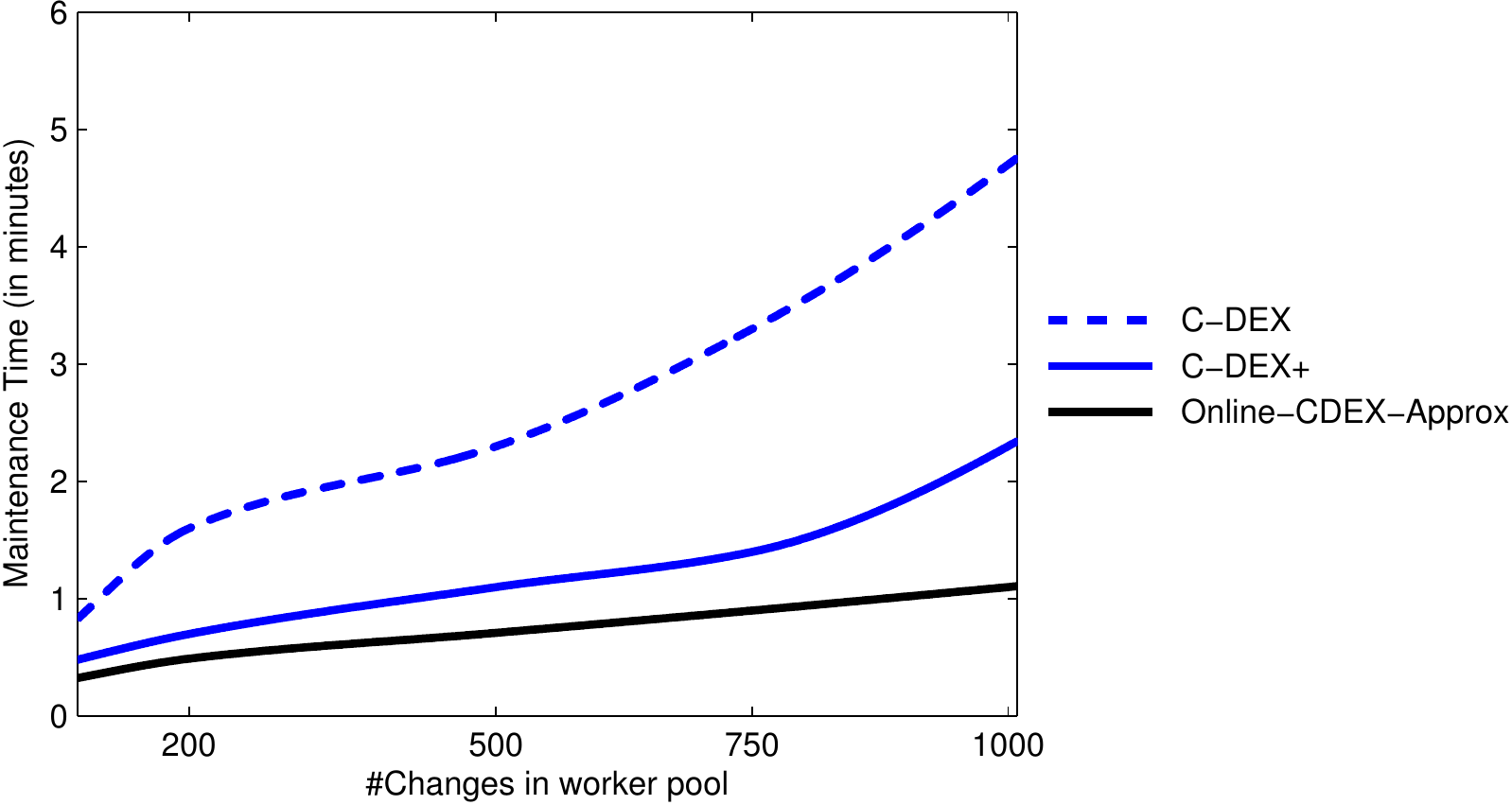}
    \caption{\small{Time for index maintenance varying \# worker addition}}
    \label{fig:15}
\end{minipage}
\end{figure*}

\vspace{-0.1in}
\section{Related Work}\label{rel}

A growing number of crowdsourcing systems are available nowadays, both as commercial platforms (like AMT and CrowdFlower) or for academic use. Examples of applications include sentence translation, photo tagging and sentiment analysis, but also query answering (CrowdDB \cite{DBLP:journals/pvldb/FengFKKMRWX11}, Qurk \cite{qurk}, Deco \cite{deco}, sCOOP, FusionCOMP, MoDaS, CyLog/Crowd4U), or entity resolution(such as CrowdER \cite{DBLP:journals/pvldb/WangKFF12}), planning queries \cite{tovaVLDB2013}, perform matching \cite{DBLP:conf/sigmod/WangLKFF13}, or counting \cite{countcrowd}. A common element shared across the above crowdsourcing systems is that the tasks that they handle are micro-task/binary. As such, the handling of these tasks does not necessitate collaboration, but plurality optimization. According to this, many workers are appointed to each micro-task, in order to identify the task's ``true value", by means of majority voting \cite{Ipeirotis:2010:QMA:1837885.1837906} or more sophisticated techniques. The optimization problem in that case is to select the correct workers to identify the true values efficiently, with as low cost as possible. Conversely, commercial systems typically allow workers to self-appoint themselves to tasks, and then apply worker filtering (based on reputation mechanisms, screening mechanisms \cite{Downs:2010:YPG:1753326.1753688}, pre-qualification tests, or ``golden data'' \cite{Josang:2007:STR:1225318.1225716}) as a means of ensuring task quality. Another means of passive quality assurance deals with refining task quality evaluation after the tasks are
completed \cite{whitehill}, or being completed
\cite{adityaReliableWorkers, DBLP:conf/sigmod/GuoPG12, DBLP:conf/icde/BoimGMNPT12}. Very recent research studies try to actively improve plurality optimization through mechanisms that suggest tasks to workers \cite{journals/corr/abs-1110-3564, DBLP:conf/icde/BoimGMNPT12,DBLP:conf/sigmod/ParameswaranGPPRW12}. Apart from plurality optimization, other optimization problems examined by current literature aim at improving the application's response time \cite{journals/corr/abs-1204-2995, Liu:2012:CCD:2336664.2336676} for micro-task/binary crowdsourcing.

Knowledge-intensive crowdsourcing (KI-C) \cite{Kittur:2013:FCW:2441776.2441923} handles tasks related to knowledge production, such as article writing, decision-making, science journalism. These tasks require a ``collaboration" among workers rather than their voting. 
Our problem bears some resemblance with existing team formation problems in social networks (SN)~\cite{www12}, in the sense that here too users are grouped together with the purpose of collaboration on a set of tasks. There are however two critical differences: whereas SN-based team formation relies on user affinity within the social network, crowdsourcing entails a huge scale of diverse worker pool unknown to each other, who do not necessarily need the {\em synergy} of a ``team'' to work together (e.g., a Wikipedia-style of work can be used).
Second, KI-C deals with unique challenges related to human factors in a dynamic environment, which is rarely seen for SN-based team formation. 

Although recent works \cite{Kulkarni:2012:CCW:2145204.2145354} acknowledge that more sophisticated methods of crowd coordination and optimization are needed to handle tasks that are knowledge-intensive, no work to the best of our knowledge does so. Our contribution is one of the first ever attempts to address this gap. 

\vspace{-0.15in}
\section{Conclusion}\label{conc}
We propose \sys, a unified framework for optimizing worker-to-task assignment in knowledge intensive crowdsourcing. \sys\ formalizes the optimization objective and designs principled optimal and approximate solutions considering multiple skills and cost, which is flexible enough to be adapted to different applications.  Unlike existing works, \sys\ makes a deliberate acknowledgement of human factors in designing the solutions. \sys\ relies on a set of pre-computed indexes, and uses them adaptively to enable effective worker-to-task assignment. The uniformity  is illustrated in handling different scenarios with appropriate adaptations. Finally, the effectiveness of \sys\ is validated through extensive real-data and synthetic experiments, considering both quality and performance.
\vspace{-0.05in}

\vspace{-0.15in}
\def\thebibliography#1{
  \section*{References}
 \small
  \list
    {[\arabic{enumi}]}
    {\settowidth\labelwidth{[#1]}
     \leftmargin\labelwidth
     \parsep 1pt                
     \itemsep 0.6pt               
     \advance\leftmargin\labelsep
     \usecounter{enumi}
    }
  \def\newblock{\hskip .11em plus .33em minus .07em}
  \sloppy\clubpenalty10000\widowpenalty10000
  \sfcode`\.=1000\relax
}
\bibliographystyle{abbrv}
\bibliography{paperbib,paperbib1}  

\end{document}